\documentclass[english,12pt]{article}
\usepackage[latin1]{inputenc}
\usepackage{babel}
\usepackage{ngerman}
\usepackage{amsmath,amssymb, mathrsfs, dsfont, amsthm, color}
\usepackage[margin=1in]{geometry}
\usepackage{tikz}

\usepackage{amsxtra}
\usepackage{stmaryrd}

\renewenvironment{thebibliography}[1]{%
\begin{oldthebibliography}{#1}%
\setlength{\baselineskip}{.9em}
\linespread{.9}
\small
\setlength{\parskip}{0ex}%
\setlength{\itemsep}{.1em}%
}%
{%
\end{oldthebibliography}%
}

\usepackage{amsfonts}
\usepackage{amsthm}
\usepackage{amsmath}
\usepackage{amssymb}
\usepackage{bbm}
\usepackage{dsfont}
\newtheorem{thm}{Theorem}[section]
\newtheorem{defi}[thm]{Definition}
\newtheorem{prop}[thm]{Proposition}
\newtheorem{cor}[thm]{Corollary}
\newtheorem{lemma}[thm]{Lemma}

\theoremstyle{definition}
\newtheorem{remark}[thm]{Remark}
\newtheorem{example}[thm]{Example}

\newcommand{\bt}{\begin{thm}}
\newcommand{\et}{\end{thm}}
\newcommand{\br}{\begin{remark}}
\newcommand{\er}{\end{remark}}
\newcommand{\bl}{\begin{lemma}}
\newcommand{\el}{\end{lemma}}
\newcommand{\bp}{\begin{proof}}
\newcommand{\ep}{\end{proof}}
\newcommand{\bal}{\begin{align*}}
\newcommand{\eal}{\end{align*}}
\newcommand{\bi}{\begin{itemize}}
\newcommand{\be}{\begin{equation}}
\newcommand{\ee}{\end{equation}}
\newcommand{\bea}{\begin{eqnarray}}
\newcommand{\eea}{\end{eqnarray}}
\newcommand{\ba}{\begin{align*}}
\newcommand{\ea}{\end{align*}}
\newcommand{\ei}{\end{itemize}}

\DeclareMathOperator{\esssup}{ess\ sup}

\DeclareMathOperator{\conv}{\mathrm{conv}}

\newcommand{\R}{\mathbb{R}}

\newcommand{\N}{\mathbb{N}}

\newcommand{\F}{\mathcal{F}}
\newcommand{\cF}{\mathcal{F}}

\newcommand{\cZ}{\mathcal{Z}}

\newcommand{\cA}{\mathcal{A}}

\newcommand{\Om}{\Omega}

\newcommand{\hg}{\widehat g}
\newcommand{\hy}{\widehat y}
\newcommand{\hX}{\widehat X}

\newcommand{\cS}{\mathcal{S}}
\newcommand{\cM}{\mathcal{M}}

\newcommand{\vt}{\vartheta}
\newcommand{\hvt}{\widehat\vartheta}
\newcommand{\tvp}{\widetilde\varphi}
\newcommand{\hvp}{\widehat\varphi}

\newcommand{\hQ}{\widehat{Q}}

\newcommand{\vp}{\varphi}
\newcommand{\vr}{\varrho}

\newcommand{\cD}{\mathcal{D}}

\newcommand{\E}{\mathcal{E}}

\newcommand{\tS}{\widetilde S}
\newcommand{\hS}{\widehat S}

\newcommand{\sint}{\stackrel{\mbox{\tiny$\bullet$}}{}}
\newcommand{\hh}{\widehat h}
\newcommand{\hZ}{\widehat Z}

\newcommand{\cC}{\mathcal{C}}

\numberwithin{equation}{section}

\newcommand{\Lim}{\lim\limits}

\begin{document}
\title{Portfolio optimisation beyond semimartingales: shadow prices and fractional Brownian motion\footnote{We would like to thank Junjian Yang for careful reading of the manuscript and pointing out a mistake in an earlier version.}}
\author{Christoph Czichowsky\footnote{Department of Mathematics, London School of Economics and Political Science, Columbia House, Houghton Street, London WC2A 2AE, UK, {\tt c.czichowsky@lse.ac.uk}. Financial support by the Swiss National Science Foundation (SNF) under grant PBEZP2\_137313 and by the European Research Council (ERC) under grant FA506041 is gratefully acknowledged.} 
\hspace{20pt}Walter Schachermayer\footnote{Fakult\"at f\"ur Mathematik, Universit\"at Wien, Oskar-Morgenstern-Platz 1, A-1090 Wien, {\tt walter.schachermayer@univie.ac.at}. Partially supported by the Austrian Science Fund (FWF) under grant P25815, the European Research Council (ERC) under grant FA506041 and by the Vienna Science and Technology Fund (WWTF) under grant MA09-003.}
\\
\\
This version: \today.
}
\date{}
\maketitle

\begin{abstract}
\noindent
While absence of arbitrage in frictionless financial markets requires price processes to be semimartingales, non-semimartingales can be used to model prices in an arbitrage-free way, if proportional transaction costs are taken into account. In this paper, we show, for a class of price processes which are not necessarily semimartingales, the existence of an optimal trading strategy for utility maximisation under transaction costs by establishing the existence of a so-called shadow price. This is a semimartingale price process, taking values in the bid ask spread, such that frictionless trading for that price process leads to the same optimal strategy and utility as the original problem under transaction costs. Our results combine arguments from convex duality with the stickiness condition introduced by P. Guasoni. They apply in particular to exponential utility and geometric fractional Brownian motion. In this case, the shadow price is an It\^o process. As a consequence we obtain a rather surprising result on the pathwise behaviour of fractional Brownian motion: the trajectories may touch an It\^o process in a one-sided manner without reflection.
\end{abstract}
\vspace{-0.2cm}\newline
\vspace{-0.2cm}\newline


\section{Introduction}
Most of the literature in mathematical finance assumes that discounted prices $S=(S_t)_{0 \leq t \leq T}$ of risky assets are modelled by \emph{semimartingales}. In frictionless financial markets, where arbitrary amounts of stock can be bought and sold at the same price $S_t$,  the semimartingale assumption is necessary. Otherwise, there would exist ``arbitrage opportunities'' (see \cite{DS94}, Th.~7.2 for a precise statement) and optimal strategies for utility maximisation problems would fail to exist or yield infinite expected utility (see \cite{AI05,LZ08,KP08}).

For non-semimartingale models based on fractional Brownian motion $(B^H_t)_{t \geq 0}$ such as the \emph{fractional Black-Scholes model} $S_t=\exp(\mu t + \sigma B^H_t)$, where $\mu \in \mathbb{R},$ $\sigma > 0$ and Hurst parameter $H \in (0,1) \setminus\{\frac{1}{2}\}$, Rogers \cite{R97} and Cheridito \cite{C03} showed explicitly how to construct these arbitrage opportunities. Such models have been proposed by Mandelbrot~\cite{M63} for their natural fractal scaling behaviour and related statistical properties. They are prime examples of non-semimartingale models to start with.

While fractional models cannot be covered by the classical theory of frictionless financial markets, recent results \cite{G06, GRS08, GRS10} illustrate that this can be done in an arbitrage-free and economically meaningful way by taking (arbitrary small) proportional transaction costs into account. As has been shown by Guasoni \cite{G06}, the crucial property for the absence of arbitrage under transaction costs is that fractional Brownian motion is sticky. Conceptually, the absence of arbitrage allows to consider portfolio optimisation also for non-semimartingale price processes under transaction costs; see \cite{G02}. However, so far there have been no results on how to obtain the optimal trading strategy in non-semimartingale models.

In this paper, we address this question. For this, we investigate portfolio optimisation under transaction costs for non-semimartingale price processes satisfying the stickiness condition such as the fractional Black-Scholes model and utility functions $U: \mathbb{R} \to \mathbb{R}$ that are defined on the whole real line. The prime example of such a utility is exponential utility $U(x)=- \exp (-x).$ Besides the non-linearity of the wealth dynamics under transaction costs, the main difficulty is that fractional Brownian motion is neither a semimartingale nor a Markov process and therefore the standard tools from stochastic calculus are quite limited. The basic idea to overcome these issues is to use the concept of a \emph{shadow price}. This is a semimartingale price process $\widehat{S}=(\widehat{S}_t)_{0 \leq t \leq T}$ such that the solution to the frictionless utility maximisation problem for this price process gives the same optimal strategy and utility as the original problem under transaction costs. 

Our {\it main result} is established in Theorem \ref{mt3} below. It shows the existence of shadow prices for utility functions that are bounded from above, under the assumption that the price process $S=(S_t)_{0 \leq t \leq T}$ is continuous and sticky. Theorem \ref{mt3} also ensures that an optimal trading strategy exists. In the frictionless case one typically assumes the existence of an equivalent local martingale measure for the price process having suitable integrability properties to achieve this. In contrast, our sufficient conditions under transaction costs are more robust and hold in a wide variety of models; see \cite{BS10,C08,GSV11,GRS08,GR15,HPR14,P11,P10}\footnote{Note that, if a process has conditional full support, it is also sticky.}. They apply in particular to the fractional Black-Scholes model and exponential utility. Moreover, we give an example that illustrates that the condition that the price process $S=(S_t)_{0 \leq t \leq T}$ is sticky cannot be replaced by the assumption that it satisfies the condition $(NFLVR)$ of ``no free lunch with vanishing risk'' (without transaction costs).

The connection to frictionless financial markets is then the key to use tools from semimartingale calculus for the potentially non-semimartingale price processes $S=(S_t)_{0 \leq t \leq T}$ by simply applying them to the shadow price process $\widehat{S}=(\widehat{S}_t)_{0 \leq t \leq T}.$ This also allows us to exploit known results for portfolio optimisation in frictionless financial markets under transaction costs. For the fractional Black-Scholes model we obtain in this manner that the shadow price process is an Itô process given by
\begin{align}\label{int:ito}
d \widehat{S}_t = \widehat{S}_t (\widehat{\mu}_t dt + \widehat{\sigma}_t dW_t), \quad 0 \leq t \leq T,
\end{align}
where $\mu=(\widehat{\mu}_t)_{0 \leq t \leq T}$ and $\widehat{\sigma}=(\widehat{\sigma}_t)_{0 \leq t \leq T}$ are predictable processes such that the solution to \eqref{int:ito} is well-defined in the sense of Itô integration. 

We expect that analysing the coefficients $\widehat{\mu}=(\widehat{\mu}_t)_{0 \leq t \leq T}$ and $\widehat{\sigma} = (\widehat{\sigma}_t)_{0 \leq t \leq T}$ of the It\^o process \eqref{int:ito} should also allow to obtain quantitative results for the optimal strategy under transaction costs in the fractional Black-Scholes model. A thorough analysis of these coefficient processes is left to future research.

By the definition of the shadow price, the optimal strategies under transaction costs and the corresponding frictionless problem only trade, if the shadow price $\hS$ is equal to the bid price or ask price, i.e.~$\hS=(1-\lambda)S$ or $\hS=S$, respectively. For sufficiently small transaction costs, we show the intuitively obvious fact that -- with high probability -- the optimal strategy actually does trade as opposed to just keeping the initial position in bond. As a consequence we obtain a rather surprising result on the pathwise behaviour of fractional Brownian motion: the trajectories may touch an It\^o process in a one-sided manner without reflection. The set on which the paths touch contains the set on which the optimal strategies trade.

It is tempting to conjecture that the above described touching of the trajectories of the fractional Brownian motion and the It\^o process happens on a Cantor-like compact subset of $[0,T]$ without isolated points and that the optimal trading strategy is continuous on $(0,T)$ and of local time type. When $S$ is the usual Black-Scholes model, it is well known that these properties hold true; see \cite{TKA88,DN90,SS94}. However, in the present fractional case, the question seems to be completely open.

The conditions that the price process $S=(S_t)_{0 \leq t \leq T}$ is continuous and sticky are invariant under equivalent changes of measure. Therefore, our main result also ensures the existence of exponential utility indifference prices for all bounded European contingent claims $C$ by the usual change of measure given by $\frac{dP_C}{dP}=\frac{\exp(C)}{E[\exp(C)]}$; compare \cite{ER00,Detal02} for the frictionless case. The question is then, if this allows to obtain more reasonable prices in the fractional Black-Scholes model. Recall that the concept of super-replication leads by the face-lifting theorem \cite{GRS08} only to economically trivial prices in these models; compare also \cite{SSC95}.

It is ``folklore'' that the existence of a shadow price is in general related to the solution of a dual problem; see \cite{KMK11,CMKS14,CS14,CSY14}.  We establish this relation for utility functions taking finite values on the whole real line and c\`adl\`ag price processes and provide the necessary duality results in this setup. Similarly as in the frictionless case \cite{S01}, this builds up upon results from utility maximisation for utility functions $U:(0,\infty)\to\R$ that have been recently established in \cite{CS14} under transaction costs. Moreover, we use an ``abstract version'' of the duality for utility functions on the whole real line in the spirit of those in \cite{KS99} for utility functions on the positive half-line.

The understanding of the duality, sometimes called the ``martingale method'', in the context of portfolio optimisation goes back to \cite{KLS87,HP91,KLSX91} in the frictionless case. Under transaction costs, our work complements the dynamic duality results \cite{CK96,CW01,CS14,CSY14} for utility functions on the positive half-line as well as the static duality results \cite{DPT01,B02,BM03,CO11,BY13} for (possibly) multi-variate utility functions.

The insight that utility maximisation can be studied under proportional transaction costs also for non-semimartingale price processes goes back to Guasoni \cite{G02}. In that paper, utility functions $U:(0,\infty)\to\R$ are considered under the assumption that the problems are well posed. However, in this setup it is not clear whether or not this assumption is satisfied for non-semimartingale processes such as the fractional Black-Scholes model and popular utilities like logarithmic utility $U(x)=\log(x)$. In particular, a counter-example in \cite{CSY14} shows that it is not sufficient to suppose that the price process is continuous and sticky to guarantee the existence of a shadow price. For utility functions $U:(0,\infty)\to\R$ that are bounded from above like power utility $U(x)=\frac{1}{p}x^p$ with $p\in(-\infty,0)$, Guasoni's result \cite{G02} applies and establishes the existence of an optimal trading strategy under transaction costs. It remains as an open question, whether a shadow price exists in this setting.\footnote{Note added in proof: This question has been answered in \cite{CPSY16} in the meantime. If the indirect utility is finite, it is sufficient for the existence of a shadow price that the price process is continuous and satisfies the condition $(TWC)$ of ``two way crossing''; see \cite{Ben12,Pey15}. Combining this with the fact that fractional Brownian motion satisfies a law of iterated logarithm not only at deterministic times but also stopping times (see Theorem 1.1 of \cite{Pey15}), it allows to deduce the existence of a shadow price for the fractional Black-Scholes model and \emph{all} utility functions $U:(0,\infty)\to\R$ satisfying the condition of reasonable asymptotic elasticity.\label{footnote}}

The paper is organised as follows. We formulate the problem in Section 2. Section 3 contains the duality results and the relation of the solution to the dual problem and the shadow price for utility functions on the whole real line. Our main result, which asserts the existence of a shadow price, is established in Section 4. We explain how to specialise this result to the fractional Black-Scholes model and exponential utility in Section 5. In Theorem~\ref{thm}, we give the result on the pathwise behaviour of fractional Brownian motion. Finally, the Appendix contains an ``abstract version'' of the duality result established in Section 3 that is used in its proof.

\section{Formulation of the problem}

We consider a financial market consisting of one riskless and one risky asset. The riskless asset is assumed to be constant to one. Trading the risky asset incurs proportional transaction costs $\lambda \in (0,1)$. This means that one has to pay a (higher) ask price $S_t$ when buying risky shares but only receives a lower bid price $(1-\lambda)S_t$ when selling them. Here \mbox{$S=(S_t)_{0\leq t\leq T}$} denotes a strictly positive, adapted, c\`adl\`ag (right-continuous process with left limits) process defined on some underlying filtered probability space $\big(\Om,\F,(\F_t)_{0\leq t\leq T},P\big)$ with fixed finite time horizon $T\in(0,\infty)$ satisfying the usual assumptions of right-continuity and completeness. As usual equalities and inequalities between random variables hold up to $P$-nullsets and between stochastic processes up to $P$-evanescent sets.

\emph{Trading strategies} are modelled by $\R^2$-valued, predictable processes $\vp=(\vp^0_t,\vp^1_t)_{0\leq t\leq T}$ of finite variation, where $\vp^0_{t}$ and $\vp^1_{t}$ describe the holdings in the riskless and the risky asset, respectively, after rebalancing the portfolio at time $t$.  For any process $\psi=(\psi_t)_{0\leq t\leq T}$ of finite variation we denote by $\psi=\psi_0+\psi^{\uparrow}-\psi^{\downarrow}$ its Jordan-Hahn decomposition into two non-decreasing processes $\psi^{\uparrow}$ and $\psi^{\downarrow}$ both null at zero. The total variation $|\psi|_t$ of $\psi$ on $(0,t]$ is then given by $|\psi|_t=\psi^{\uparrow}_t+\psi^{\downarrow}_t$. For $0\leq s<t\leq T$, the total variation of $\psi$ on $(s,t]$ denoted by $\int_s^t|d\psi_u|$ is then simply $\int_s^t|d\psi_u|=|\psi|_t-|\psi|_s$. Note that, any process $\psi$ of finite variation is in particular l\`adl\`ag (with right and left limits). For any l\`adl\`ag process $X=(X_t)_{0\leq t\leq T}$, we denote by $X^c$ its continuous part given by
$$X^c_t:=X_t-\sum_{s<t} \Delta_+ X_s -  \sum_{s\leq t} \Delta X_s,$$
where $\Delta_+ X_t:=X_{t+}-X_t$ are its right and $\Delta X_t:=X_t-X_{t-}$ its left jumps.

A strategy $\vp=(\vp^0_t,\vp^1_t)_{0\leq t\leq T}$ is called \emph{self-financing}, if 
\begin{equation}\label{eq:sf}
\int_s^td\varphi^0_u\leq-\int_s^tS_ud\varphi^{1,\uparrow}_u +\int_s^t (1-\lambda)S_ud\varphi^{1,\downarrow}_u, \quad 0 \leq s\leq  t \leq T,
\end{equation}
where the integrals
\begin{align*}
\int^t_s S_u d\varphi^{1,\uparrow}_u&:= \int^t_s S_u d\varphi_u^{1,\uparrow,c} + \sum_{s < u \leq t}S_{u-}\Delta \varphi_u^{1,\uparrow} + \sum_{s \leq u < t} S_u \Delta_+ \varphi_u^{1,\uparrow},\\
\int^t_s(1-\lambda) S_u d\varphi_u^{1,\downarrow}&:= \int^t_s (1-\lambda) S_u d\varphi_u^{1, \downarrow,c} + \sum_{s < u \leq t} (1-\lambda)S_{u-}\Delta\varphi^{1, \downarrow}_u + \sum_{s \leq u < t} (1-\lambda)S_u \Delta_+ \varphi^{1,\downarrow}_u
\end{align*}
can be defined pathwise by using Riemann-Stieltjes integrals, as explained in \cite{CS13,CS14,S13} for example. The total variation of $\vp=(\vp^0,\vp^1)$ on $(s,t]$ is given by $\int_s^t|d\vp_u|=\int_s^t|d\vp^0_u|+\int_s^t|d\vp^1_u|$.

A self-financing strategy $\vp=(\vp^0,\vp^1)$ is called \emph{admissible}, if there exists some constant $M>0$ such that its \emph{liquidation value} satisfies
\be
V^{liq}_t(\vp):=\vp^0_t+(\vp^1_t)^+(1-\lambda)S_t-(\vp^1_t)^-S_t\geq -M, \quad 0 \leq t \leq T.\label{liq}
\ee

For $x\in\R$, we denote by $\cA^{\lambda}_{adm}(x)$ the set of all self-financing and admissible trading strategies under transaction costs $\lambda$ starting from initial endowment $(\vp^0_0,\vp^1_0)=(x,0)$ and
$$\cC^\lambda_{b}(x):=\{V^{liq}_T(\vp)~|~\vp=(\vp^0,\vp^1)\in\cA^{\lambda}_{adm}(x)\}.$$

As explained in Remark 4.2 in \cite{CS06}, we can assume without loss of generality that $\vp^1_T=0$ and therefore have
$$\cC^\lambda_{b}(x)=\{\vp^0_T~|~\vp=(\vp^0,\vp^1)\in\cA^{\lambda}_{adm}(x)\}.$$

 A \emph{$\lambda$-consistent price system} is a pair of stochastic processes $Z=(Z^0_t, Z^1_t)_{0 \leq t \leq T}$ consisting of the density process $Z^0=(Z^0_t)_{0 \leq t \leq T}$ of an equivalent local martingale measure $Q\sim P$ for a price process $\widetilde{S}=(\widetilde{S}_t)_{0 \leq t \leq T}$ evolving in the bid-ask spread $[(1-\lambda)S,S]$ and the product $Z^1=Z^0\widetilde{S}$. Requiring that $\widetilde{S}$ is a local martingale under $Q$ is tantamount to the product $Z^1=Z^0\widetilde{S}$ being a local martingale under $P$. Similarly, an \emph{absolutely continuous $\lambda$-consistent price system} is a pair of stochastic processes $Z=(Z^0_t, Z^1_t)_{0 \leq t \leq T}$ consisting of the density process $Z^0=(Z^0_t)_{0 \leq t \leq T}$ of an absolutely continuous local martingale measure $Q\ll P$ for a price process $\widetilde{S}=(\widetilde{S}_t)_{0 \leq t \leq T}$ evolving in the bid-ask spread $[(1-\lambda)S,S]$ and the product $Z^1=Z^0\widetilde{S}$ which is assumed to be a local martingale. Under transaction costs these concepts play a similar role as equivalent and absolutely continuous local martingale measures in the frictionless case. We denote by $\cZ^\lambda_e$ the set of all $\lambda$-consistent price systems and by $\cZ^\lambda_a$ the set of all absolutely continuous $\lambda$-consistent price systems.
 
 
While absence of arbitrage in the frictionless setting in the form of the existence of an equivalent local martingale measure for the price process $S=(S_t)_{0 \leq t \leq T}$ implies that it has to be a semimartingale (this property is invariant under equivalent changes of measure), non-semimartingales can be used to model asset prices in an arbitrage-free way as soon as proportional transaction costs are taken into account. Indeed, for the prime example of a non-seminarmartingale, geometric franctional Brownian motion $S_t:= \exp(B^H_t)$ with Hurst parameter $H \in (0,1) \setminus \{\frac{1}{2}\}$, Guasoni \cite{G06} showed that this price process is arbitrage-free for any proportion $\lambda \in (0,1)$ of transaction costs and hence admits a $\lambda$-consistent price system for any $\lambda \in (0,1)$ by the fundamental theorem of asset pricing for continuous processes under small transaction costs in \cite{GRS10}. As has been observed by Guasoni, the crucial property of fractional Brownian motion, which allows to deduce the arbitrage freeness, is that it is sticky.

\begin{defi}\label{def:sticky}
A stochastic process $X=(X_t)_{0 \leq t \leq T}$ is \emph{sticky}, if
$$P\left(\sup_{t \in [\tau, T]} | X_t - X_\tau | < \delta,\, \tau < T\right) > 0,$$
for any $[0,T]$-valued stopping time $\tau$ with $P(\tau<T)>0$ and any $\delta > 0.$
\end{defi}

By Proposition 2 in \cite{BS10} the stickiness condition is preserved under a transformation of the process $X=(X_t)_{0 \leq t \leq T}$ by continuous functions. Therefore it does not make a difference, if we require that the $\mathbb{R}_+$-valued process $S=(S_t)_{0 \leq t \leq T}$ or $X_t:= \log (S_t)$ is sticky.

In this paper, we want to investigate the existence of optimal trading strategies in models based on fractional Brownian motion $(B^H_t)$ such as the \emph{fractional Black-Scholes model}, where
$$S_t =\exp (\mu t + \sigma B^H_t), \quad 0 \leq t \leq T,$$
where $\mu \in \mathbb{R}$ and $\sigma>0$.

To that end, we consider a utility function $U:\mathbb{R} \to \mathbb{R}$ that is defined and finite on the whole real line, increasing, strictly convex, continuously differentiable and satisfying the Inada conditions $U'(-\infty) = \lim_{x \to -\infty} U'(x)=\infty$ and $U'(\infty)= \lim_{x \to \infty} U'(x)=0.$ The prime example of such a utility is exponential utility $U(x)=-\exp(-x)$. While for utility functions on the positive half-line negative wealth is forbidden by the admissibility condition of non-negative wealth, this is not ruled out in the present setting but only penalised by giving it a low utility. Therefore, the optimal trading strategy is in general not attained in the set of admissible trading strategies (which are uniformly bounded from below) and the ``good definition'' of ``allowed'' trading strategies becomes crucial; see \cite{S03} for results in the frictionless setting. In the frictionless case, there are two approaches to deal with this issue. The first is to use a dual definition and to consider all trading strategies whose wealth processes are a super-martingale under all equivalent local martingale measures (ELMM) $Q$ with finite $V$-expectation, i.e.~$E[V(y\frac{dQ}{dP})]<\infty$ for some $y>0$, where $V(y):=\sup_{x\in\R}\{U(x)-xy\}$ for $y>0$ denotes the Legendre transform of $-U(-x)$; see for example \cite{Detal02,KS02,BF08}.

We follow the second approach of \cite{S01} to consider the ``closure'' of the set of terminal wealths of admissible trading strategies with respect to expected utility.

For this, we define 
\begin{multline*}
\cC^\lambda_U (x) = \big\{ g \in L^0 (P; \mathbb{R} \cup \{\infty\})~ |~ \exists g_n \in \cC^\lambda_b (x)~\text{s.t.} \quad \text{$U(g_n) \in L^1(P)$ and $U(g_n) \stackrel{L^1(P)}{\longrightarrow} U(g)$} \big\}
\end{multline*}
and consider the maximisation problem
\begin{equation} \label{P1}
E[U(g)] \to \max!, \quad g \in \cC_U^{\lambda}(x).
\end{equation}
Clearly,
\begin{align}\label{tag5}
u(x):=\sup_ {g \in \cC^\lambda_U(x)}E[U(g)]=\sup_ {g \in \cC^\lambda_b(x)}E[U(g)].
\end{align}

Note that $U(g_n) \stackrel{L^1(P)}{\longrightarrow} U(g)$ implies that $g_n \to g$ in $L^0(\Om, \mathcal{F}, P; \mathbb{R} \cup \{\infty\}),$ with respect to convergence in probability, since $U:\mathbb{R} \to \mathbb{R}$ is strictly increasing.

While the $g_n$ are real-valued random variables, it may -- a priori -- indeed happen that the solution $\hg(x)$ to \eqref{P1} takes the value $\infty$ with strictly positive probability, i.e.~$P(\hg(x) = \infty) > 0.$ As explained in \cite{A05} in the frictionless case, this can only happen, if $U(\infty) < \infty$, and does not contradict the no arbitrage assumption. In our setting under transaction costs, we show in Example \ref{Ex} below how this phenomenon arises. The question is therefore: does there exist a self-financing trading strategy $\hvp=(\hvp^0_t, \hvp^1_t)_{0 \leq t \leq T}$ under transaction costs $\lambda$ that attains the solution $\hg(x)$ to \eqref{P1} in the sense that $\hg(x)= V_T^{liq} (\hvp)$? For this, we consider the set $\mathcal{A}^\lambda_U (x)$ of all predictable finite variation processes $\varphi=(\varphi^0_t, \varphi^1_t)_{0 \leq t \leq T}$, starting at $(\varphi^0_0, \varphi^1_0)=(x,0),$ satisfying the self-financing condition \eqref{eq:sf} and such that there exists $ \varphi^n=(\varphi^{0,n}, \varphi^{1,n}) \in \mathcal{A}^\lambda_{adm} (x)$ verifying that $U\big(V^{liq}_T (\varphi^n)\big)\in L^1(P)$ and $U\big(V^{liq}_T (\varphi^n)\big) \stackrel{L^1(P)}{\longrightarrow}U\big(V^{liq}_T (\varphi)\big).$

Note that the latter convergence again implies that $V^{liq}_T (\varphi^n) \stackrel{L^0(P)}{\longrightarrow}V^{liq}_T (\varphi)$ by the strict monotonicity of $U$.

Requiring only that the terminal liquidation value $V^{liq}_T (\varphi)$ can be approximated by the terminal liquidation values $V^{liq}_T (\varphi^n)$ of admissible trading strategies $ \varphi^n=(\varphi^{0,n}, \varphi^{1,n}) \in \mathcal{A}^\lambda_{adm} (x)$ seems to be a rather weak version of attainability. However, as we shall see in Proposition \ref{mt2} and Theorem \ref{mt3} below, our results yield that
$$P\left[(\varphi^{0,n}_t,\varphi^{1,n}_t)\rightarrow (\varphi^0_t, \varphi^1_t) ,\ \forall t\in[0,T]\right] = 1,$$
which implies
$$P\left[V^{liq}_t(\varphi^n)\rightarrow V^{liq}_t(\varphi) ,\ \forall t\in[0,T]\right] = 1$$
by the definition of the liquidation value in \eqref{liq}.

We investigate the question of attainability by using the concept of a shadow price.

\begin{defi}\label{def:sp}
A semimartingale price process $\widehat{S}=(\widehat{S}_t)_{0 \leq t \leq T}$ is called a \emph{shadow price process}, if 
\bi
\item[\bf{1)}] $\widehat{S}$ is valued in the bid-ask spread $[(1-\lambda) S, S]$ 
\item[\bf{2)}] The solution $\hvt=(\hvt_t)_{0 \leq t \leq T}$ to the frictionless utility maximisation problem
\begin{equation}\label{P3}
E[U(x + \vt \sint \widehat{S}_T)]\to \max!, \quad \vt \in \mathcal{A}_U(x;\hS),
\end{equation}
exists in the sense of \cite{S01}. Here, $\mathcal{A}_U(x;\hS)$ denotes the set of all $\widehat{S}$-integrable (in the sense of Itô), predictable processes $\vt=(\vt_t)_{0 \leq t \leq T}$ such that there exists a sequence $(\vt^n)_{n=1}^\infty$ of self-financing and admissible trading strategies $\vt^n=(\vt^n_t)_{0 \leq t \leq T}$ without transaction costs\footnote{That is $\hS$-integrable, predictable processes $\vt^n=(\vt^n_t)_{0 \leq t \leq T}$ such that
$X_t=x+\vt^n\sint \hS_t\geq - M(n)$  for all $0\leq t\leq T$
for some constant $M(n)>0$ that might depend on $n$; see \cite{S01} for example.} such that $U(x + \vt^n \sint \widehat{S}_T)\in L^1(P)$ and $U(x + \vt^n \sint \widehat{S}_T) \stackrel{L^1(P)}{\longrightarrow} U(x + \vt \sint \widehat{S}_T).$ 
\item[\bf{3)}] The optimal trading strategy $\hvt=(\hvt_t)_{0 \leq t \leq T}$ to the frictionless problem \eqref{P3} coincides with the holdings in stock $\hvp^1=(\hvp^1_t)_{0 \leq t \leq T}$ to the utility maximisation problem \eqref{P1} under transaction costs such that $x+\hvt\sint\hS_T=V^{liq}_T(\hvp)=\hg(x)$.

\end{itemize}
\end{defi}

The basic idea is that, if a shadow price $\widehat{S}=(\widehat{S}_t)_{0 \leq t \leq T}$ for \eqref{P1} exists, this allows us to obtain the optimal trading strategy for the utility maximisation  problem \eqref{P1} under transaction costs by solving the frictionless utility maximisation problem \eqref{P3}. To the frictionless problem \eqref{P3}, we can then apply all known results from the frictionless theory to solve it. Since the shadow price $\widehat{S}=(\widehat{S}_t)_{0 \leq t \leq T}$ has to be a semimartingale, this allows us in particular to transfer some of the techniques from semimartingale calculus to utility maximisation problem \eqref{P1} for the possible non-semimartingale price process $S=(S_t)_{0 \leq t \leq T}.$

Note that the existence of a shadow price implies that the optimal strategy $\hvt=(\hvt_t)_{0 \leq t \leq T}$ to the frictionless problem \eqref{P3} is of finite variation and that both optimal strategies $\hvt=(\hvt_t)_{0 \leq t \leq T}$ and $\hvp^1=(\hvp^1_t)_{0 \leq t \leq T}$ that coincide $\hvt=\hvp^1$ only trade, if $\widehat{S}$ is at the bid or ask price, i.e. 
$$\{d\hvt=d\hvp^1>0\}\subseteq\{\hS=S\}  \quad \mbox{and} \quad \{d\hvt=d\hvp^1<0\}\subseteq\{\hS=(1-\lambda)S\}$$
in the sense that
\begin{align}
\{d\hvt^c=d\hvp^{1,c}>0\}&\subseteq \{\hS=S\}, & \{d\hvt^c=d\hvp^{1,c}<0\}&\subseteq \{\hS=(1-\lambda)S\},\notag \\
\{\Delta \hvt=\Delta \hvp>0\}&\subseteq \{\hS_-=S_-\}, &  \{\Delta \hvt=\Delta \hvp<0\}&\subseteq \{\hS_-=(1-\lambda)S_-\}, \notag \\
\{\Delta_+ \hvt=\Delta_+ \hvp>0\}&\subseteq \{\hS=S\}, &  \{\Delta_+ \hvt=\Delta_+ \hvp<0\}&\subseteq \{\hS=(1-\lambda)S\}.\label{2.10}
\end{align}

Here, a precise mathematical meaning of the inclusions \eqref{2.10} above is given by
\begin{align*}
\int^T_0 \mathbbm{1}_{\{\hS\ne S\}}(u)\hvp^{1,\uparrow}_u={}&\int^T_0 \mathbbm{1}_{\{\hS\ne S\}}(u) d\hvp_u^{1,\uparrow,c}+ \sum_{0 < u \leq T}\mathbbm{1}_{\{\hS_-\ne S_-\}}(u)\Delta \hvp_u^{1,\uparrow}\\
&+ \sum_{0 \leq u < T} \mathbbm{1}_{\{\hS\ne S\}}(u)\Delta_+ \hvp_u^{1,\uparrow}=0,\\
\int^T_0 \mathbbm{1}_{\{\hS\ne(1-\lambda)S\}}(u)\hvp^{1,\downarrow}_u={}& \int^T_0 \mathbbm{1}_{\{\hS\ne(1-\lambda)S\}}(u) d\hvp_u^{1,\downarrow,c}+ \sum_{0 < u \leq T}\mathbbm{1}_{\{\hS_-\ne (1-\lambda)S_-\}}(u)\Delta \hvp_u^{1,\downarrow}\\
&+ \sum_{0 \leq u < T} \mathbbm{1}_{\{\hS\ne (1-\lambda)S\}}(u) \Delta_+ \hvp_u^{1,\downarrow}=0.
\end{align*}

It is ``folklore'' that the shadow price is related to the solution of the dual problem; see Proposition 3.9 of \cite{CS14} for example. In the present setting of a utility function that is defined on the whole real line, we explain this relation in the next section.

\section{Duality theory} 
We discuss the connections between shadow prices and the solution to the dual problem for utility functions on the whole real line. The following duality relations can be obtained similarly as their frictionless counterparts in \cite{S01}. This has already been implicitly exploited in the static setup of \cite{B02}. We will prove this result in the appendix by reducing it to an ``abstract version''.
\bt[Utility functions on the whole real line]\label{mainthm}
Suppose that $S$ is locally bounded and admits a $\lambda'$-consistent price system for all $\lambda'\in(0,1)$, that $U:\R\to\R$ satisfies the Inada conditions $U'(-\infty) = \lim_{x \to -\infty} U'(x)=\infty$ and $U'(\infty)= \lim_{x \to \infty} U'(x)=0$, has reasonable asymptotic elasticity, i.e. $AE_{\infty}(U):=\varlimsup\limits_{x\to\infty}\frac{xU'(x)}{U(x)}<1$ and $AE_{-\infty}(U):=\varliminf\limits_{x\to-\infty}\frac{xU'(x)}{U(x)}>1$, and that
\be
u(x):=\sup_{g\in\cC^\lambda_U(x)}E[U(g)]<U(\infty)\label{eq:fu}
\ee
for some $x\in\R$. Then:
\bi
\item[{\bf 1)}] The primal value function $u$, defined in \eqref{tag5}, and the dual value function
$$v(y):=\inf_{(Z^0,Z^1)\in\mathcal{Z}^\lambda_{a}}E[V(yZ^0_T)],$$
where $V(y):=\sup_{x\in\R}\{U(x)-xy\}$ for $y>0$ denotes the Legendre transform of $U$, are conjugate, i.e.,
\begin{eqnarray*}u(x)=\inf_{y>0}\{v(y)+xy\},\qquad v(y)=\sup_{x\in\R}\{u(x)-xy\},
\end{eqnarray*}
and continuously differentiable. The functions $u$ and $-v$ are strictly concave and satisfy the Inada conditions
$$\text{$\Lim_{x\to-\infty}u'(x)=\infty,\qquad\Lim_{y\to\infty}v'(y)=\infty,\qquad\Lim_{x\to\infty}u'(x)=0,\qquad\Lim_{y\to0}v'(y)=-\infty$}.$$
The primal value function $u$ has reasonable asymptotic elasticity. 
\item[{\bf 2)}] For $y>0$, the solution $\hZ(y)=\big(\hZ^0(y),\hZ^1(y)\big)\in\mathcal{Z}^\lambda_{a}$ to the dual problem
\begin{equation}
\textstyle
E\left[V\big(yZ^0_T\big)\right]\to\min!\label{D1}, \qquad{Z=(Z^0,Z^1)\in\mathcal{Z}^\lambda_{a}},
\end{equation}
exists, the first component $\hZ^0_T(y)$ is unique and the map $y\mapsto\hZ^0_T(y)$ is continuous in variation norm.
\item[{\bf 3)}] For $x\in\R$, the solution $\hg(x)\in\cC^\lambda_U(x)$ to the primal problem \eqref{P1} exists, is unique and given by
\be
\hg(x)=(U')^{-1}\left(\hy(x)\hZ^0_T\big(\hy(x)\big)\right),\label{eq:dr:1}
\ee
where $\hy(x)=u'(x)$.

\item[{\bf 4)}] We have the formulae
$$v'(y)=E\left[\hZ^0_T(y)V'\big(y\hZ^0_T(y)\big)\right]\qquad\text{and}\qquad xu'(x)=E\left[\hg(x)U'\big(\hg(x)\big)\right],$$
where we use the convention that $0\cdot\infty=0$, if the random variables are of this form.

\ei
\et

Why did we focus on utility functions $U$ taking finite values on the entire real line? The reason is that, for utility functions $U\!:(0, \infty) \to \mathbb{R}$ on the positive half-line shadow prices might fail to exist due to the fact that the solution to the dual problem is not necessarily attained as a local martingale but in general only as a supermartingale; see for example \cite{BCKMK13,CMKS14,CS14,CSY14}. We do not know how to successfully overcome this difficulty for models like the fractional Black--Scholes model in that context.\footnote{Note added in proof: We answered this question in \cite{CPSY16} in a quite satisfactory way: for the fractional Black-Scholes model, there exists a shadow price for {all} utility functions $U:(0,\infty)\to\R$ satisfying the condition of reasonable asymptotic elasticity. See the footnote on page \pageref{footnote} for more details.}
 This ``supermartingale phenomenon'' does not appear for utilities $U:\mathbb{R} \to \mathbb{R}$ on the whole real line, the dual optimiser is guaranteed to be a local martingale. On the other hand, the solution $\widehat{Z}=(\widehat{Z}^0, \widehat{Z}^1)$ to the dual problem \eqref{D1} may  -- in general -- fail to induce a shadow price due to the fact that it might only be a absolutely continuous $\lambda$-consistent price system, i.e.~that $P(\widehat{Z}^0_T =0) > 0.$ By the duality relation \eqref{eq:dr:1}, the set $\{\widehat{Z}^0_T=0\}$ is equal to the set $\{\hg(x)=\infty\}.$ As $V(0)=U(\infty)$, such a behaviour can only arise, if $U(\infty) < \infty$ and there exists no $\lambda'$-consistent price system $\bar{Z}=(\bar{Z}^0_t, \bar{Z}^1_t)_{0 \leq t \leq T}$ such that $E[V(y\bar{Z}^0_T)] < \infty$ for some $y > 0$; compare \cite{C75, S01, A05} for the frictionless case. For utility functions such that $U(\infty)=\infty$, the dual optimiser $\widehat{Z}=(\widehat{Z}^0, \widehat{Z}^1)$, provided it exists, always satisfies $\hZ^0_T>0$ almost surely. However, for these utility functions, the condition \eqref{eq:fu} seems hard to verify for non-semimartingale price process such as the fractional Black-Scholes model.

The following proposition shows that the existence of a strictly consistent price system with finite $V$-expectation ensures the attainability of the primal optimiser $\hg(x)$. It generalises Lemma 25 in \cite{B02} to our setting and its proof follows by similar arguments. 
\begin{prop}\label{mt2}
Under the assumptions of Theorem \ref{mainthm}, suppose that, for some $\lambda' \in (0,\lambda)$, there exists a $\lambda'$-consistent price system $\bar{Z}=(\bar{Z}^0, \bar{Z}^1) \in \cZ^{\lambda'}_e$ such that
$$E[V(\bar{y}\bar{Z}^0_T)] < \infty$$
for some $\bar{y}>0.$ Then the solution to the primal problem \eqref{P1} is attainable, i.e.~there exists $\hvp=(\hvp^0,\hvp^1)\in\cA^\lambda_{U}(x)$ such that $V^{liq}_T(\hvp)=\hg(x)$, and there exist $\tvp^n=(\tvp^{0,n},\tvp^{1,n})\in\cA^\lambda_{adm}(x)$ such that
\be
P\left[(\widetilde{\varphi}^{0,n}_t,\widetilde{\varphi}^{1,n}_t)\rightarrow (\widehat{\varphi}^0_t, \widehat{\varphi}^1_t) ,\ \forall t\in[0,T]\right] = 1\label{mt2:eq1}.
\ee
\end{prop}
\begin{proof}
By Theorem \ref{mainthm} there exists a sequence $\vp^n=(\vp^{0,n},\vp^{1,n})\in\cA^\lambda_{adm}(x)$ such that
$$U\big(V^{liq}_T(\vp^n)\big)\xrightarrow{L^1(P)}U\big(\hg(x)\big).$$ Then $\left(\bar{Z}^0_t\left(\vp^{0,n}_t+\vp^{1,n}_t\bar{S}_t+A^n_t\right)\right)_{0\leq t\leq T}$
is a supermartingale for all $n$, where $\bar{S}:=\frac{\bar{Z}^1}{\bar{Z}^0}$ and $A^n_t :=(\lambda-\lambda') \int_0^tS_ud\vp^{1,n,\downarrow}_u$. Indeed, by integration by parts we can write
$$\bar{Z}^0_t (\vp^{0,n}_t+\vp^{1,n}_t\bar{S}_t)= \bar{Z}^0_t \left(\vp^{0,n}_t + \int_0^t\bar{S}_u d\vp^{1,n}_u+ \int_0^t\vp_u^{1,n}d\bar{S}_u \right).$$
Since $\bar{S}\in[(1-\lambda')S, S]$ and
 $$\vp^{0,n}_t\leq x-\int_0^tS_ud\vp^{1,n,\uparrow}_u+\int_0^t(1-\lambda)S_ud\vp^{1,n,\downarrow}_u$$
 by the self-financing condition \eqref{eq:sf}, the process
$ \left(\vp^{0,n}_t + \int_0^t\bar{S}_u d\vp^{1,n}_u+A^n_t\right)_{0\leq t\leq T}$ is non-increasing. Moreover, by Bayes' rule $\bar{S}$ is a local martingale under the measure $\bar{Q}\sim P$ given by $\frac{d\bar{Q}}{dP}=\bar{Z}^0_T$ and, since $\vp^{1,n}$ is of finite variation and hence locally bounded, the stochastic integral $\vp^{1,n}\sint \bar{S}$ is a local martingale under $\bar{Q}$. Therefore $\bar{Z}^0\left(\vp^{0,n}+\vp^{1,n}\bar{S}+A^n\right)$ is a local supermartingale under $P$ again by Bayes' rules that is bounded from below by $\bar{Z}^0V^{liq}(\vp^n)$. Since $\vp^n\in\cA^\lambda_{adm}(x)$ is admissible and $\bar{Z}^0$ a martingale, this implies that $\bar{Z}^0\left(\vp^{0,n}+\vp^{1,n}\bar{S}+A^n\right)$ is a true supermartingale so that
\begin{equation}
   E\left[\bar{Z}^0_T\left(V^{liq}_T(\vp^n)+A^n_T\right) \right]=E\left[\bar{Z}^0_T\left(\vp^{0,n}_T+(\lambda-\lambda')\int_0^TS_ud\vp^{1,n, \downarrow}_u\right) \right]\leq x\label{l:att:eq1}
 \end{equation}
 for all $n$. Combining Fenchel's inequality with the monotonicity of $U$ we can estimate
 $$\bar{Z}^0_T\left(V^{liq}_T(\vp^n)+A^n_T\right)\geq \frac{1}{\bar{y}}\left(U\big(V^{liq}_T(\vp^n)\big)-V\big(\bar{y}\bar{Z}^0_T\big)\right).$$
 Since $\frac{1}{\bar{y}}\left(U\big(V^{liq}_T(\vp^n)\big)-V\big(\bar{y}\bar{Z}^0_T\big)\right)\xrightarrow{L^1(P)}\frac{1}{\bar{y}}\left(U\big(\hg(x)\big)-V\big(\bar{y}\bar{Z}^0_T\big)\right)$, as $n\to\infty$, we obtain that $\left(\bar{Z}^0_T\left(V^{liq}_T(\vp^n)+A^n_T\right)^-\right)_{n=1}^\infty$ is uniformly integrable and hence that $\left(\bar{Z}^0_T\left(V^{liq}_T(\vp^n)+A^n_T\right)\right)_{n=1}^\infty$ is bounded in $L^1(P)$ by \eqref{l:att:eq1}. Since $\bar{Z}^0_T>0$ and $V^{liq}_T(\vp^n)\overset{L^0(P)}{\longrightarrow}\hg(x)$, this implies that $\conv\big\{A^n_T~;~n\geq1\big\}$
is bounded in $L^0(P)$. Since $\bar{S}$ is as a non-negative local $\bar{Q}$-martingale also a $\bar{Q}$-supermartingale, we have that $\inf_{0\leq u\leq T}S_u\geq \inf_{0\leq u\leq T}\bar{S}_u>0$ by the minimum principle for supermartingales. This implies that $\conv\big\{|\vp^{1,n}|_T~;~n\geq1\big\}$ and hence $\conv\big\{|\vp^{0,n}|_T~;~n\geq1\big\}$ are bounded in $L^0(P)$ as well. By Proposition 3.4 in \cite{CS06} (and its application in the proof of Theorem 3.5 therein) there exists a sequence
\begin{equation*}
 (\widetilde{\varphi}^{0,n},\widetilde{\varphi}^{1,n})\in\conv\left((\varphi^{0,n},\varphi^{1,n}), (\varphi^{0,n+1},\varphi^{1,n+1}), \dots \right)
\end{equation*}
of convex combinations and a predictable process $\hvp=(\widehat{\varphi}^0_t, \widehat{\varphi}^1_t)_{0\leq t\leq T}$ of finite variation such that 
\begin{equation}
 P\left[(\widetilde{\varphi}^{0,n}_t,\widetilde{\varphi}^{1,n}_t)\rightarrow (\widehat{\varphi}^0_t, \widehat{\varphi}^1_t) ,\ \forall t\in[0,T]\right] = 1.\label{l:att:eq2}
\end{equation}
The convergence \eqref{l:att:eq2} then implies that $\hvp=(\widehat{\varphi}^0, \widehat{\varphi}^1)$ is a self-financing trading strategy under transaction costs $\lambda$ such that $V^{liq}_T(\hvp)=\hg(x)$ and hence $\hvp\in\cA_U^\lambda(x)$.
\end{proof}
The next result shows that \eqref{mt2:eq1} is sufficient to guarantee the existence of a shadow price.
\begin{prop}\label{mt2B}
Under the assumptions of Theorem \ref{mainthm}, suppose that the solution $\hg(x)$ to the primal problem \eqref{P1} is attainable, i.e.~there exists $\hvp=(\hvp^0,\hvp^1)\in\cA^\lambda_{U}(x)$ such that $V^{liq}_T(\hvp)=\hg(x)$, and that there exist $\tvp^n=(\tvp^{0,n},\tvp^{1,n})\in\cA^\lambda_{adm}(x)$ such that
\be
P\left[(\widetilde{\varphi}^{0,n}_t,\widetilde{\varphi}^{1,n}_t)\rightarrow (\widehat{\varphi}^0_t, \widehat{\varphi}^1_t) ,\ \forall t\in[0,T]\right] = 1\label{mt2:eq1B}
\ee
Then the dual optimiser $\widehat{Z}=(\widehat{Z}^0, \widehat{Z}^1)$ to \eqref{D1} is in $\cZ^{\lambda}_e,$ i.e.~ a $\lambda$-consistent price system, and $\widehat{S}:=\frac{\widehat{Z}^1}{\widehat{Z}^0}$ is a shadow price (in the sense of Definition $\ref{def:sp}$) to problem \eqref{P1}.
\end{prop}
\bp
Since $\hg(x)=V^{liq}_T(\hvp)<\infty$, we have that $\hy(x)\hZ^0_T=U'\big(\hg(x)\big)>0$ by the duality relation \eqref{eq:dr:1} and therefore that the dual optimiser $\widehat{Z}=(\widehat{Z}^0, \widehat{Z}^1)$ is in $\cZ^{\lambda}_e$. It then follows along the same arguments as in the proof of Proposition \ref{mt2} after replacing $\vp^n=(\vp^{0,n},\vp^{1,n})$ by $\tvp^n=(\tvp^{0,n},\tvp^{1,n})$ and $\bar{Z}=(\bar{Z}^0,\bar{Z}^1)$ by $\hZ=(\hZ^0,\hZ^1)$ and setting $\lambda'=\lambda$ that $(\widehat{Z}^0 \tvp^{0,n} + \widehat{Z}^1 \tvp^{1,n})_{n=1}^\infty$ is a sequence of supermartingales $\widehat{Z}^0 \tvp^{0,n} + \widehat{Z}^1 \tvp^{1,n}=(\widehat{Z}_t^0 \tvp^{0,n}_t + \widehat{Z}^1_t \tvp^{1,n}_t)_{0\leq t\leq T}$ such that $\left(\big(\widehat{Z}^0_T \tvp^{0,n}_T + \widehat{Z}^1_T \tvp^{1,n}_T\big)^-\right)_{n=1}^\infty$
is uniformly integrable. This implies that each $\left(\big(\widehat{Z}^0_t \tvp^{0,n}_t + \widehat{Z}^1_t \tvp^{1,n}_t\big)^-\right)_{0\leq t\leq T}$ is a non-negative submartingale and hence of class (D) so that $\left(\big(\widehat{Z}^0_\tau \tvp^{0,n}_\tau + \widehat{Z}^1_\tau \tvp^{1,n}_\tau\big)^-\right)_{n=1}^\infty$ is uniformly integrable for every $[0,T]$-valued stopping time $\tau$. Since
$$\text{$\widehat{Z}^0_\tau \tvp^{0,n}_\tau + \widehat{Z}^1_\tau \tvp^{1,n}_\tau\xrightarrow{\text{$P$-a.s.}}\widehat{Z}^0_\tau \hvp^0_\tau + \widehat{Z}^1_\tau \hvp^1_\tau,\quad$ as $n\to\infty$,}$$
for every $[0,T]$-valued stopping time $\tau$ by \eqref{mt2:eq1B}, we obtain that $(\widehat{Z}^0_t \hvp^0_t + \widehat{Z}^1_t \hvp^1_t)_{0\leq t\leq T}$ is a supermartingale by Fatou's lemma that has by part 4) of Theorem \ref{mainthm} constant expectation and is therefore a martingale.

By integration by parts we get that
$$\widehat{Z}^0 \hvp^0 + \widehat{Z}^1\hvp^1 = \widehat{Z}^0 (\hvp^0 + \hvp^1 \widehat{S})= \widehat{Z}^0 (x + \hvp^1 \sint \widehat{S} - A),$$
where
$$A_t= \int^t_0 \big(\widehat{S}_u - (1-\lambda)S_u\big) d\hvp^{1, \downarrow}_u + \int^t_0 \big(S_u - \widehat{S}_u\big) d \hvp_u ^{1, \uparrow},\quad 0 \leq t \leq T,$$
is a non-decreasing, predictable process.

Since $\widehat{Z}^0 \hvp^0 + \widehat{Z}^1 \hvp^1$ is a martingale and $\widehat{Z}^0 (x + \hvp^1 \sint \widehat{S})$ is a local martingale by Bayes' rule and the fact that $\hvp^1$ is of finite variation and hence locally bounded, this implies that $A\equiv 0$ and therefore that $\widehat{Z}^0(\hvp^0 + \hvp^1 \widehat{S} )= \widehat{Z}^0(x + \hvp^1 \sint \widehat{S})$ is a martingale and $\{d\hvp^1 > 0 \} \subseteq \{\widehat{S}=S\}$ and $\{ d\hvp^1 < 0 \} \subseteq \{\widehat{S}=(1-\lambda)S\}$ in the sense of \eqref{2.10}. As $\widehat{Z}=(\widehat{Z}^0, \widehat{Z}^1) \in \cZ^{\lambda}_{e},$ we obtain that $\widehat{Z}^0= (\widehat{Z}^0_t)_{0 \leq t \leq T}$ is the density process of an ELMM for the frictionless price process $\widehat{S}=(\widehat{S}_t)_{0 \leq t \leq T}.$ Therefore $\hZ^0=(\hZ^0_t)_{0\leq t\leq T}$ and $\hy(x)$ have to be also the solution to the frictionless dual problem
$$E[V(yZ_T)]+xy\to\min!,\quad y>0,\ Z\in\cZ_{a}(\hS),$$
where $\cZ_{a}(\hS)$ denotes the set of all density processes $Z=(Z_t)_{0\leq t\leq T}$ of absolutely continuous martingale measures $Q\ll P$ for the locally bounded price process $\hS=(\hS_t)_{0\leq t\leq T}$. It follows from the frictionless duality (see Theorem 2.2 in \cite{S01}) that $x + \hvp^1 \sint \widehat{S}_T=\hvp^0_T + \hvp^1_T \widehat{S}_T = V^{liq}_T (\hvp) = U' \big(\hy(x) \widehat{Z}^0_T\big)$ is the optimal terminal wealth to the frictionless utility maximisation problem \eqref{P3} for $\widehat{S}=(\widehat{S}_t)_{0 \leq t \leq T}$. Since $x + \hvp^1 \sint \widehat{S}$ is a $\hQ$-martingale under the measure $\hQ\sim P$ given by $\frac{d\hQ}{dP}=\hZ^0_T$ by Bayes' rule, we obtain that $\hvp^1=(\hvp^1_T)_{0 \leq t \leq T}$ has to be the optimal strategy to the frictionless utility maximisation problem \eqref{P3} and therefore in $\cA_U(x; \widehat{S})$ by part $(iv)$ of Theorem 2.2 in \cite{S01}, as the optimal strategy is unique in $L(\hS)$. This implies that $\widehat{S}=(\widehat{S}_t)_{0 \leq t \leq T}$ is a shadow price process in the sense of Definition \ref{def:sp} for the utility maximisation problem \eqref{P1} under transaction costs.
\ep
\section{The main result}
\bt\label{mt3}
Suppose that $S$ is continuous and sticky and that $U:\mathbb{R} \to \mathbb{R}$ is strictly concave, increasing, continuously differentiable, bounded from above, satisfying the Inada condition $U'(-\infty)=\lim_{x\to-\infty}U'(x)=-\infty$ and having reasonable asymptotic elasticity, i.e. $\varliminf_{x \to -\infty} \frac{xU'(x)}{U(x)} > 1$.

Then we have for any $x\in\R$ and any proportion of transaction costs $\lambda\in(0,1)$ that:
\bi
\item[\bf{1)}] An optimal trading strategy $\widehat{\varphi}(x) = (\widehat{\varphi}^0_t(x), \widehat{\varphi}^1_t(x)) _{0 \leq t \leq T} \in \mathcal{A}^{\lambda}_U(x)$ for \eqref{P1} exists.
\item[\bf{2)}] There exist admissible trading strategies $\tvp^n=(\tvp^{0,n},\tvp^{1,n})\in\cA^\lambda_{adm}(x)$ which are maximising for \eqref{P1} and such that 
$$P\left[(\widetilde{\varphi}^{0,n}_t,\widetilde{\varphi}^{1,n}_t)\rightarrow (\widehat{\varphi}^0_t, \widehat{\varphi}^1_t) ,\ \forall t\in[0,T]\right] = 1.$$
In fact, for every maximising sequence $(\varphi^n)^\infty_{n=1} \in \mathcal{A}^{\lambda}_{adm} (x)$ we can find a sequence $(\widetilde{\varphi}^n)^\infty_{n=1}$ of convex combinations with the above properties. 
\item[\bf{3)}]  The dual optimiser $\widehat Z=(\widehat Z^0,\widehat Z^1)$ to \eqref{D1} is in $\mathcal{Z}^\lambda_e$, i.e.~a $\lambda$-consistent price system.
\item[\bf{4)}] $\hS:=\frac{\hZ^1}{\hZ^0}$ is a shadow price (in the sense of Def.~\ref{def:sp}). This implies in particular that 
\begin{align*}
\{d\hvp^{1,c}>0\}&\subseteq \{\hS=S\}, & \{d\hvp^{1,c}<0\}&\subseteq \{\hS=(1-\lambda)S\},\notag \\
\{\Delta \hvp^1>0\}&\subseteq \{\hS_-=S\}, &  \{\Delta \hvp^1<0\}&\subseteq \{\hS_-=(1-\lambda)S	\}, \notag \\
\{\Delta_+ \hvp^1>0\}&\subseteq \{\hS=S\}, &  \{\Delta_+ \hvp^1<0\}&\subseteq \{\hS=(1-\lambda)S\}.
\end{align*}

\ei


\et
\label{mt3:p1}
The proof of Theorem \ref{mt3} will be broken into several lemmas. We begin by verifying the conditions of the duality theorem (Theorem \ref{mainthm}). Since $S$ is continuous and sticky, combining Corollary 2.1 in \cite{G06} and Theorem 2 in \cite{GRS10} yields the existence of a strictly consistent price system for all sizes of transaction costs $\lambda'\in(0,1)$. Moreover, the conditions on the utility function $U$ are satisfied by our assumptions. Therefore, we only need to check condition \eqref{eq:fu}.
\begin{lemma}\label{l:nc}
Let $U:\mathbb{R} \to \mathbb{R}$ be a utility function  that is bounded from above and $S=(S_t)_{0\leq t\leq T}$ be sticky. Then we have, for all $x\in\R$, that
\begin{equation}\label{l:finite:eq:1}
u(x)=\sup_{\varphi \in \mathcal{A}^{\lambda}_{adm} (x)} E [ U(V^{liq}_T(\varphi)) ] < U(\infty).
\end{equation}
\end{lemma}

\begin{proof}
By the stickiness of $S=(S_t)_{0\leq t\leq T}$ and hence that of $X_t:= \log (S_t)$ the set\\
$$A:=\left\{\sup_{t\in[0,T]} \left| \frac{S_0}{S_t} -1 \right| <\frac{\lambda}{3}\right\} \supseteq \left\{\sup_{t\in[0,T]} \mid X_t-X_0 \mid < \log \left(1+\frac{\lambda}{3}\right)\right\}$$
has strictly positive measure, i.e. $P[A]>0$.
Similarly as in Lemma 2.5 and Proposition 2.8 in \cite{G06} we then have that $V^{liq}_T(\varphi) \leq x$ on $A$ for any $\varphi\in \mathcal{A} ^\lambda_{adm}(x)$. Indeed, using the self-financing condition \eqref{eq:sf} under transaction costs we obtain that 
\begin{align}
V^{liq}_T(\varphi)&=\varphi^0_T + \varphi^1_T S_T - \lambda S_T (\varphi^1_T)^+\nonumber\\
&\leq x -\int\limits_{0}^{T} S_u d\varphi^1_u - \lambda \int\limits_{0}^{T} S_u d\varphi^{1,\downarrow}_u + \varphi^1_T S_T - \lambda S_T (\varphi^1_T)^+\nonumber\\
&= x -\int\limits_{0}^{T} (S_u - S_0) d\varphi^1_u - \lambda \int\limits_{0}^{T} S_u  d\varphi^{1,\downarrow}_u + \varphi^1_T (S_T-S_0) - \lambda S_T (\varphi^1_T)^+ \nonumber\\
& \leq x - \frac{2}{3} \lambda \int\limits_{0}^{T} S_u d\varphi^{1,\downarrow}_u-\frac{2}{3}\lambda S_T (\varphi^1_T)^+\leq x \quad \text{on $A$}.\label{eq:s1}
\end{align}
This implies that
$$E[U(V^{liq}_T(\varphi))] \leq U(\infty) (1-P[A]) + U(x) P[A] < U(\infty)$$
for all $\varphi \in \mathcal{A}^{\lambda}_{adm} (x)$ and therefore \eqref{l:finite:eq:1} by taking the supremum.
\end{proof}

Applying the duality theorem (Theorem \ref{mainthm}) allows us to obtain a maximising sequence $\varphi^n=(\varphi^{0,n}_t, \varphi^{1,n}_t)_{0 \leq t \leq T}\in\cA^\lambda_{adm}(x)$ of self-financing and admissible trading strategies and a random variable $\hg=\hg(x) \in L^0 (P; \mathbb{R} \cup \{\infty\})$ such that $E\big[U\big(\hg(x)\big)\big] = u(x)$ and
\begin{align}
V^{liq}_T(\varphi^n) &\stackrel{P}{\longrightarrow} \hg(x),\nonumber\\
U\big(V^{liq}_T (\varphi^n)\big)&\stackrel{L^1(P)}{\longrightarrow} U\big( \hg(x)\big).\label{ms}
\end{align}

As already mentioned it may -- a priori -- happen that the random variable $\hg(x)$ takes the value $\infty$ with strictly positive probability. The following example illustrates how this phenomenon arises under transaction costs. It shows, in particular, that the condition that $S=(S_t)_{0\leq t\leq T}$ is sticky in Theorem \ref{mt3} cannot be replaced by the assumption that $S=(S_t)_{0\leq t\leq T}$ satisfies the condition $(NFLVR)$ of ``no free lunch with vanishing risk'' (without transaction costs).

\begin{example}\label{Ex}
We give an example of a price process $S=(S_t)_{0\leq t\leq 1}$ such that

\bi
\item[{\bf 1)}] $S$ is continuous.
\item[{\bf 2)}] $S$ satisfies the condition $(NFLVR)$ without transaction costs and therefore admits a $\lambda'$-consistent price system for all $\lambda'\in(0,1)$.
\item[{\bf 3)}] There exists no optimal trading strategy to the problem of maximising exponential utility $U(x)=-\exp(-x)$ under transaction costs $\lambda\in(0,\frac{1}{2})$, that is,
$$E\big[U\big(V_1^{liq}(\vp)\big)\big]=E\big[-\exp\big(-V_1^{liq}(\vp)\big)\big]\to\max!,\quad\vp\in\cA^\lambda_U(x).$$
\item[{\bf 3')}] There exists a sequence $\hvp^n=(\hvp^{0,n}_t,\hvp^{1,n}_t)_{0\leq t\leq 1}\in\cA^\lambda_U(x)$ such that $$U\big(V_1^{liq}(\hvp^n)\big)\xrightarrow{L^1(P)}0=U(\infty)$$
and therefore $\hg(x)=\infty$ $P$-a.s. In particular, we have that $|\hvp^n|_T\xrightarrow{P}\infty$.
\ei
For convenience, we give the construction on the infinite time interval $[0,+\infty]$. The corresponding example on the finite interval $[0,1]$ can be obtained by using a time change $h:[0,+\infty]\to [0,1]$ given by $h(t)=\big(1-\exp(-t)\big)$ and  considering $S_{h(t)}$ instead of $S_t$. 

We begin by specifying the ask price $S=(S_t)_{0\leq t\leq\infty}$ under an equivalent local martingale measure $Q$. Let $W=(W_t)_{t\geq 0}$ be a Brownian motion on $[0,+\infty)$ under $Q$ and set
$$\textstyle\sigma:=\inf\{t>0~|~\E(W)_t=\exp(W_t-\frac{1}{2}t)=\frac{1}{2}\}.$$
Define $S=(S_t)_{0\leq t\leq\infty}$ by
$$S_t=2 \E(W)^{\sigma}_t, \quad 0\leq t\leq\infty.$$

In prose, the price process $S$ starts at $2$. It then fluctuates until it hits the level $1$ for the first time at time $\sigma$ and then remains constant afterwards. Since the stopping time $\sigma$ is almost surely finite, we have that the price process is a non-negative local martingale under $Q$ such that $S_\infty=1$ $Q$-a.s.

Therefore, short selling one share of stock at time $0$ yields $2(1-\lambda)-1>0$ at time $\infty$ as liquidation value.

The problem with this strategy is, of course, that it is not admissible. Since the stock price can get arbitrarily high with strictly positive probability, the liquidation value $V^{liq}(\vp)$ can get arbitrarily small with strictly positive probability between $0$ and $\sigma$. However, we can approximate this strategy by admissible trading strategies $\bar\vp^n=(\bar\vp^{0,n}_t,\bar\vp^{1,n}_t)_{0\leq t\leq \infty}\in\cA^\lambda_{adm}(0)$. For this, we simply set $\bar\vp^{1,n}_t=-\mathbbm{1}_{\rrbracket0,\sigma_n\rrbracket}(t)$ for $0\leq t\leq\infty$, where $\sigma_n:=\inf\{t>0~|~S_t=n\}$, and define $\bar\vp^{0,n}_t$ via the self-financing condition \eqref{eq:sf} with equality. Then 
\begin{align*}
V^{liq}_\infty(\bar\vp^n)={}&\big(2(1-\lambda)-1\big)\mathbbm{1}_{\{\sigma_n\geq\sigma\}}\\
&+\big(2(1-\lambda)-n\big)\mathbbm{1}_{\{\sigma_n<\sigma\}}\xrightarrow{\text{$P$-a.s.}}1+2(1-\lambda)-1,\quad\text{as $n\to\infty$},
\end{align*}
since $\sigma_n\nearrow\infty$ $Q$-a.s. Therefore, setting $\hvp^{1,n}=n\bar\vp^{1,n}$ and $\hvp^{0,n}=n\bar\vp^{0,n}$ gives a sequence $(\hvp^n)_{n=1}^\infty$ of self-financing and admissible trading strategies $\hvp^n=(\hvp^{0,n}_t,\hvp^{1,n}_t)_{0\leq t\leq \infty}\in\cA^\lambda_{adm}(0)$ such that 
\begin{align*}
U\big(V^{liq}_\infty(\hvp^n)\big)={}&-\exp\big(-n\big(2(1-\lambda)-1\big)\big)\mathbbm{1}_{\{\sigma_n\geq\sigma\}}\\&-\exp\big(-n\big(2(1-\lambda)-n\big)\big)\mathbbm{1}_{\{\sigma_n<\sigma\}}\xrightarrow{\text{$P$-a.s.}}0,\quad\text{as $n\to\infty$}.
\end{align*}
To ensure the convergence also in $L^1(P)$, we need to specify the distribution of $S$ under $P$. Since 
\begin{align*}
E\big[U\big(V^{liq}_\infty(\hvp^n)\big)\big]={}&-\exp\big(-n\big(2(1-\lambda)-1\big)\big)P(\sigma_n\geq\sigma)\\&-\exp\big(-n\big(2(1-\lambda)-n\big)\big)P(\sigma_n<\sigma)
\end{align*}
and $-\exp\big(-\big(1+n\big(2(1-\lambda)-n\big)\big)\big)=O\big(\exp(n^2)\big)$, it will be sufficient to choose $P\sim Q$ such that $P(\sigma_n<\sigma)=o\big(\exp(-n^2)\big)$. This is possible because $A_n:=\{\sigma_n<\sigma\}$ is a decreasing sequence of sets such that $Q(A_n)>0$ and $Q(A_n)\searrow0$.

To obtain $u(x)<U(\infty)$, we flip a fair coin at time $0$. If head shows up, we use the above price process. If we observe tail, then the price process stays at $2$.
\end{example}
The above example indicates that $\hg(x)$ can only take the value $\infty$, if the total variations $(|\hvp^n|_T)_{n=1}^\infty$ of the maximising sequence $\varphi^n=(\varphi^{0,n}_t, \varphi^{1,n}_t)_{0 \leq t \leq T}\in\cA^\lambda_{adm}(x)$ of admissible trading strategies diverge to $\infty$. However, this behaviour leads to an infinite amount of trading volume and therefore of transaction costs. This cannot be optimal for a sticky price process and we now argue how to exclude it. For this, we observe that, if we have that
\begin{align}\label{4:eq1}
C:=\conv\{|\varphi^n|_T~; ~n \geq 1\} 
\end{align}
is bounded in $L^0(P)$ for a sequence $(\varphi^n)^\infty_{n=1}$ of strategies $\varphi^n \in \mathcal{A}^{\lambda}_{adm}(x)$ satisfying \eqref{ms}, there exists a sequence $(\hvp^n)^\infty_{n=1}$ of convex combinations
$$\hvp^n \in \conv (\bar{\varphi}^n, \bar{\varphi}^{n+1}, \dots)$$
and a self-financing trading strategy $\hvp=(\hvp^0_t, \hvp^1_t)_{0 \leq t \leq T}$ under transaction costs such that
\begin{align}\label{4:eq2}
P\left[(\hvp^{0,n}_t, \hvp^{1,n}_t) \xrightarrow{n\to \infty}(\hvp^0_t, \hvp^1_t),\,\forall t \in [0,T]\right]=1
\end{align}
by Proposition 3.4 in \cite{CS06} (and its application in the proof of Theorem 3.5 therein).\footnote{Note that, since $C\subseteq L^0_+(P)$ is convex and bounded, there exists by, for example, Lemma 2.3 in \cite{BS99} a probability measure $Q\sim P$ such that $C$ is bounded in $L^1(Q)$ so that the sequence $(\bar{\varphi}^n)_{n=1}^\infty$, indeed, satisfies the assumptions of Proposition 3.4 in \cite{CS06}.}

Since we then have in particular
\begin{align*}
V^{liq}_T(\hvp^n) &\stackrel{P}{\longrightarrow} V^{liq}_T (\hvp)=\hg(x)\\
U\big( V^{liq}_T (\hvp^n)\big) &\stackrel{L^1(P)}{\longrightarrow} U\big(V_T^{liq}(\hvp)\big)=U\big(\hg(x)\big),
\end{align*}
this implies that $\hvp=(\hvp^0_t, \hvp^1_t)_{0 \leq t \leq T}\in\cA^\lambda_{U}(x)$ attains the solution $\hg(x)$ to \eqref{P1} and that $\hg(x)$ is a.s.~real-valued. Therefore, it only remains to show that $(\varphi^n)^\infty_{n=1}$ satisfies \eqref{4:eq1} which will be true for any sequence $(\varphi^n)^\infty_{n=1}$ of strategies $\varphi^n \in \mathcal{A}^{\lambda}_{adm}(x)$ satisfying \eqref{ms}.

To that end, we fix any sequence $(\varphi^n)^\infty_{n=1}$ of strategies $\varphi^n \in \mathcal{A}^{\lambda}_{adm}(x)$ satisfying \eqref{ms} and denote by $\cS$ the set of all $[0,T] \cup \{\infty\}$-valued stopping times $\sigma$ such that 
$$\conv\{|\varphi^n|_{\sigma \wedge T}~;~n \geq 1\}$$
is bounded in $L^0(P).$ Then \eqref{4:eq1} corresponds to showing that $\infty \in \cS.$

\begin{lemma}
The set $\cS$ is \emph{stable under taking pairwise maxima}, i.e.~$\sigma_1, \sigma_2 \in \cS$ implies $\sigma_1 \vee \sigma_2 \in \cS.$
\end{lemma}

\begin{proof}
Let $\psi \in A:= \conv\{|\varphi^n|~;~n \geq 1\}$. Then
$$\psi_{(\sigma_1 \vee \sigma_2) \wedge T}=\psi_{\sigma_1 \wedge T}\mathbbm{1}_{\{\sigma_1\geq \sigma_2\}}+\psi_{\sigma_2 \wedge T}\mathbbm{1}_{\{\sigma_1< \sigma_2\}}.$$
This implies that 
\begin{align*}
\lim_{N \to \infty} \sup_{\psi \in A} P\left(\psi_{(\sigma_1 \vee \sigma_2) \wedge T} \geq N\right) \leq \lim_{N \to \infty} \sup_{\psi \in A} P\left(\psi_{\sigma_1 \wedge T} \geq N\right) + \lim_{N \to \infty} \sup_{\psi \in A} P\left(\psi_{\sigma_2 \wedge T} \geq N\right)=0
\end{align*}
and hence that $\sigma_1 \vee \sigma_2 \in \cS.$
\end{proof}

The fact that $\cS$ is stable under taking pairwise maxima allows us to obtain its essential supremum 
\begin{align}\label{806}
\widehat{\sigma}:=\esssup_{\sigma \in \cS} \sigma
\end{align}
as a limit of an increasing sequence $(\widehat{\sigma}_k)^\infty_{k=1}$ of stopping times $\widehat{\sigma}_k \in \cS$ by Theorem A.33.(b) in \cite{FS11}. Note that $\widehat{\sigma} \geq 0,$ as $0 \in \cS$, and that $\widehat{\sigma}$ again is a stopping time.

Recall that the existence of a shadow price implies that the optimal trading strategy $\hvp=(\hvp^0_t,\hvp^1_t)_{0 \leq t \leq T}$ under transaction costs only trades, if the shadow price is at the bid or ask price in the sense of \eqref{2.10}. The next lemma shows that this is already the case in an approximate sense, if we do not yet know, whether or not there is a shadow price.
\bl\label{l:B1}
Under the assumptions of Theorem \ref{mt3}, let $(\vp^n)_{n=1}^\infty$ be a maximising sequence of admissible trading strategies $\varphi^n=(\varphi^{0,n}_t, \varphi^{1,n}_t)_{0 \leq t \leq T}\in\cA^\lambda_{adm}(x)$ for problem \eqref{P1} satisfying \eqref{ms} and set $B_{1,j}=\{ \hZ^0 S-\hZ^1>\frac{1}{j}\}$ and $B_{2,j}=\{ \hZ^1-\hZ^0(1-\lambda) S>\frac{1}{j}\}$ for $j\in\N$. Then we have, for all $j\in\N$, that
\begin{align*}
&\mathbbm{1}_{B_{1,j}}\sint\vp^{1,n,\uparrow}_T+\mathbbm{1}_{B_{2,j}}\sint\vp^{1,n,\downarrow}_T\xrightarrow{P}0,\\
&\mathbbm{1}_{B_{1,j}}\sint\vp^{0,n,\downarrow}_T+\mathbbm{1}_{B_{2,j}}\sint\vp^{0,n,\uparrow}_T\xrightarrow{P}0.
\end{align*}
\el
\bp
Here, we can without loss of generality assume that we have equality in the self-financing condition \eqref{eq:sf} for the maximising strategies $(\vp^n)_{n=1}^\infty$. Since $0<\sup_{0\leq t\leq T}S_t<\infty$ $P$-a.s. by the assumption that $S$ is strictly positive and continuous, it is sufficient to prove the assertion for $\varphi^{1,n}=(\varphi^{1,n}_t)_{0 \leq t \leq T}$. This implies the assertion as well for $\varphi^{0,n}=(\varphi^{0,n}_t)_{0 \leq t \leq T}$ by the self-financing condition \eqref{eq:sf}.

By Lemma \ref{l:A}, we have that 
$$\widehat{Z}^0_T \vp^{0,n}_T \stackrel{L^1(P)}{\longrightarrow} \widehat{Z}^0_T \hg(x)$$
for any maximising sequence $\varphi^n=(\varphi^{0,n}_t, \varphi^{1,n}_t)_{0 \leq t \leq T}$ of self-financing and admissible trading strategies satisfying \eqref{ms}. As we can without loss of generality assume that $\varphi^{1,n}_T=0$, defining 
$$\widehat{X}^n_t = \varphi^{0,n}_t \widehat{Z}^0_t + \varphi^{1,n}_t \widehat{Z}^1_t, \quad 0 \leq t \leq T,$$
gives a sequence $(\widehat{X}^n)^\infty_{n=1}$ of supermartingales $\widehat{X}^n=(\widehat{X}^n_t)_{0 \leq t \leq T}$ starting at $x$ such that $\widehat{X}^n_T$ converges in $L^1(P)$ to the terminal value $\widehat{X}^\infty_T=\widehat{Z}^0_T \widehat{g}(x)$ of the martingale $\widehat{X}^\infty=(\widehat{X}^\infty_t)_{0 \leq t \leq T}$ given by
$$\widehat{X}^\infty_t=E[\widehat{Z}^0_T \hg(x) | \mathcal{F}_t], \quad 0 \leq t \leq T,$$
that is also starting at $x$ by part 4) of Theorem \ref{mainthm}. 

By integration by parts, we obtain
$$\widehat{X}^n_t = x + \varphi^{0,n} \sint \hZ^0_t+\varphi^{1,n} \sint \hZ^1_t - A^n_t,\quad 0 \leq t \leq T,$$
where
\begin{align*}
A^n_t:=&\int^t_0 \big(\hZ^0_uS_u  - \hZ^1_u \big) d\varphi_u^{1,n,\uparrow}+ \int^t_0 \big(\hZ^1_u -\hZ^0_u(1-\lambda) S_u\big) d\varphi^{1,n,\downarrow}_u, \quad 0 \leq t \leq T,
\end{align*}
 is a non-decreasing process starting at $0$. Since 
$$\hX^n_t=\widehat{Z}^0_t \big(\varphi^{0,n}_t + \varphi^{1,n}_t \widehat{S}_t\big) \geq \hZ^0_tV^{liq}_t (\varphi^n)\geq \hZ^0_t(- m), \quad 0 \leq t \leq T,$$
for some $m>0$ by the admissibility of $\vp^n$, the local martingale \mbox{$(x + \varphi^{0,n} \sint \hZ^0_t+\varphi^{1,n} \sint \hZ^1_t)_{0\leq t \leq T}$} is bounded from below by the uniformly integrable martingale $\big(\hZ^0_t(- m)\big)_{0\leq t \leq T}$ and hence a supermartingale. As the supermartingales $\widehat{X}^n=(\widehat{X}^n_t)_{0 \leq t \leq T}$ and the martingale $\widehat{X}^\infty=(\widehat{X}^\infty_t)_{0 \leq t \leq T}$ are both starting at $x$, the convergence $\hX^n_T\xrightarrow{L^1(P)}\hX^{\infty}_T$ therefore implies that $A^n_T\xrightarrow{L^1(P)}0$. Since
$$A^n_T\geq \frac{1}{j}\left(\mathbbm{1}_{B_{1,j}}\sint\vp^{1,n,\uparrow}_T+\mathbbm{1}_{B_{2,j}}\sint\vp^{1,n,\downarrow}_T\right)\geq0,$$
the latter $L^1$-convergence yields that  $\mathbbm{1}_{B_{1,j}}\sint\vp^{1,n,\uparrow}_T+\mathbbm{1}_{B_{2,j}}\sint\vp^{1,n,\downarrow}_T\xrightarrow{L^1(P)}0$ and hence also in probability.
\ep
We establish the following lemma to prove that $\widehat{\sigma}$ as defined in \eqref{806} equals $\widehat{\sigma}=\infty$ by contradiction. 

\begin{lemma}\label{A:l1}
Under the assumptions of Theorem \ref{mt3}, suppose that $P(\widehat{\sigma} < \infty)>0.$ Then there exists a stopping time $\tau$ with $P(\tau<T)>0$ such that we have
\bi
\item[{\bf 1)}] $\conv\{|\varphi^n|_{\tau \wedge T}~;~n \geq 1\}$
is bounded in $L^0(P),$
\item[{\bf 2)}] there exists a set $A\in\cF$ with $A\subseteq\{\tau<T\}$ and $P(A) >0,$ a constant $c>0$ and a sequence $(\hvp^n)^\infty_{n=1}$ of convex combinations
$$\hvp^n \in \conv (\varphi^n, \varphi^{n+1}, \dots)$$
such that we have on $A$ that
\bi
\item[{\bf a)}] $\int^\tau_0 |d\hvp^n_u| \leq c$ for all $n$,
\item[{\bf b)}] $\int^T_\tau|d\hvp^n_u| \xrightarrow{P} \infty,$ as $n \to \infty$,
\item[{\bf c)}] $|S_t - S_\tau| \leq \frac{\lambda}{3} S_t$ for all $t \in [\tau,T].$
\ei
\ei
\end{lemma}

\begin{proof}
Let $X_t=\log (S_t)$ and define the stopping time
$$\varrho:= \inf\Big\{t > \widehat{\sigma}~\Big|~|X_t - X_{\widehat{\sigma}}| > \textstyle\frac{1}{3} \log \big(1+ \frac{\lambda}{3}\big)\Big\}.$$
Clearly, $\vr>\widehat{\sigma}$ on $\{\widehat{\sigma}<T\}$ so that $P(\vr>\widehat{\sigma})=P(\widehat{\sigma} < T) > 0$. Hence
\begin{align}\label{A:eq4}
D:=\conv\{|\varphi^n|_{\varrho \wedge T}~;n \geq 1\}
\end{align}
is not bounded in $L^0(P)$ by the definition of $\widehat{\sigma}.$ Moreover, since $D\subseteq L^0_+(P)$ is convex, there exists by Lemma 2.3 in \cite{BS99} a partition of $\Om$ into disjoint sets $\Om_u, \Om_b \in \mathcal{F}_{\varrho}$ with $P(\Om_u) > 0$ such that 
\bi
\item[{\bf (i)}]
The restriction $D|_{\Om_b}=\{g\mathbbm{1}_{\Om_b}~|~g\in D\}$ of $D$ to $\Om_b$ is bounded in $L^0(P).$
\item[{\bf (ii)}] $D$ is hereditarily unbounded in $L^0(P)$ on $\Om_u$. That is, for every subset $B\in\cF$, $B\subseteq\Om_u$, $P(B)>0$, we have that $D|_{B}=\{g\mathbbm{1}_{B}~|~g\in D\}$ fails to be bounded in $L^0(P)$; see Definition 2.2 in \cite{BS99}.
\ei
Now, we can have two cases. Either $P(\Om_u \cap \{\varrho \geq T\}) > 0$ or $P(\Om_u \cap \{\varrho < T\})=P(\Om_u).$ In the first case, we set $F:=\Om_u \cap \{\varrho \geq T\}.$ In the second one, there exists by the stickiness of $S$ and hence that of $X$ a set $F \in \mathcal{F} $ with $P(F) > 0$ such that $F \subseteq \Om_u \cap \{\varrho < T\}$ and $\sup_{t \in [\varrho,T]}|X_t - X_{\varrho} | < \frac{1}{3} \log (1+ \frac{\lambda}{3})$ on $F.$

By the continuity of $S$, we can choose $k \in \mathbb{N}$ sufficiently large such that
$$\sup_{t \in [\widehat{\sigma}_k, \widehat{\sigma}]} | X_t - X_{\widehat{\sigma}_k} | < \frac{1}{3} \log \Big(1 + \frac{\lambda}{3}\Big)$$
on a set $A\in\cF$ with $A\subseteq F$ and $P(A) > 0.$

Setting $\tau=\widehat{\sigma}_k$, we then have 1) by \eqref{806} and that
$$\sup_{t \in [\tau, T]}|X_t - X_\tau | < \log \left(1 + \frac{\lambda}{3}\right)\quad\text{on $A$},$$
which implies that 
$$\text{$|S_t - S_\tau| \leq \frac{\lambda}{3} S_t$ for all $t \in [\tau,T]$ on $A$.}$$
By part 4) of Lemma 2.3 in \cite{BS99}, assertion (ii) above yields the existence of a sequence $(\psi^n)^\infty_{n=1}$ of convex combinations
\be
\psi^n\in\conv\{|\varphi^m|~;~m \geq n\}\label{psi}
\ee
such that
\begin{align}\label{A:eq5}
P\Big(\Om_u \cap \big\{\psi^n_{\varrho \wedge T} < n \big\} \Big) < \frac{1}{n}.
\end{align}
Since $\conv\{|\varphi^n|_{\tau\wedge T}~;~n \geq 1\}$ is bounded in $L^0(P)$, we can by an application of Koml\'os' lemma (see, for example, Lemma A.1 in \cite{DS94}) assume without loss of generality that
\be
\psi^n_{\tau\wedge T}\xrightarrow{\text{$P$-a.s.}}f,\quad\text{as $n\to\infty$},\label{psi1}
\ee
for some $f\in L^0_+(P)$.

Let $(\hvp^n)_{n=1}^\infty$ be a sequence of convex combinations 
$$\hvp^n=\sum_{k=1}^{K_n}\mu_k^n\vp^{m^n_k}\in\conv(\vp^n,\vp^{n+1},\ldots)$$
that is obtained from the sequence $(\varphi^n)_{n=1}^\infty$ by taking the same convex weights that lead to the sequence $(\psi^n)^\infty_{n=1}$ in \eqref{psi} from the sequence $(|\varphi^n|)_{n=1}^\infty$. By \eqref{psi1} and the convexity of the total variation, we can assume by possibly passing to a smaller set $A$ that still has positive probability $P(A)>0$ that there exists a constant $c >0$ such that 
$$|\hvp^n|_{\tau \wedge T} \leq c \quad \mbox{for all} \quad n \in \mathbb{N} \quad \mbox{on} \quad A.$$
This proves properties a) and c) of part 2). 

To establish property b), we need to consider the following two cases:
\bi
\item[\bf{(i')}] $P(\hZ^0_{\varrho\wedge T}=0,\ A)>0$,
\item[\bf{(ii')}] $P(\hZ^0_{\varrho\wedge T}>0,\ A)>0$.
\ei

In case (i'), it follows from the fact that $\hZ^0=(\hZ^0_t)_{0\leq t\leq T}$ is a non-negative martingale that $G:=\{\hZ^0_{\varrho\wedge T}=0\}\subseteq\{\hZ^0_T=0\}$. By the duality relation $\hg(x)=(U')^{-1}\big(\hy(x)\hZ^0_T\big)$, this implies that $\hg(x)=\lim_{n\to\infty}V^{liq}_T(\hvp^n)=\infty$ on $G$. Since $S=(S_t)_{0\leq t\leq T}$ is strictly positive and continuous, we have that $0<\sup_{0\leq t\leq T}S_t<\infty$ $P$-a.s. The only way we can have $\hg(x)=\lim_{n\to\infty}V^{liq}_T(\hvp^n)=\infty$ on $G$ is therefore that $\lim_{n\to\infty}|\hvp^n|_T=\infty$ on $G$ by
\begin{align*}
V^{liq}_T(\hvp^n)&\leq x -\int\limits_{0}^{T} S_u d\hvp^{1,n,\uparrow}_u +  \int\limits_{0}^{T}(1-\lambda) S_u d\hvp^{1,n,\downarrow}_u + \hvp^{1,n}_T S_T - \lambda S_T (\hvp^{1,n}_T)^+\\
 &\leq x +\left(\sup_{0\leq t\leq T}S_t\right)  | \hvp^{1,n}|_T  \to \infty \quad \mbox{on} \quad G.
\end{align*}
As $| \hvp^n|_{\tau\wedge T}\leq c$ for all $n\geq 1$ on $A\subseteq G$, we have that $\int^T_\tau|d\hvp^n_u| \to \infty$ on $\{\hZ^0_{\varrho\wedge T}=0\}\cap A$.

In case (ii'), we need to show that the fact that the sequence $(\psi^n)_{n=1}^\infty$ of convex combinations of total variation processes is unbounded in $L^0(P)$ in the sense of \eqref{psi1} implies that the sequence $(|\hvp^n|)_{n=1}^\infty$ of total variations of convex combinations is unbounded in $L^0(P)$ in the same sense. While this is not true in general, it follows in the present situation from the fact that all trading strategies $\hvp^n=(\hvp^{0,n}_t, \hvp^{1,n}_t)_{0 \leq t \leq T}$ of any maximising sequence satisfying \eqref{ms} have to buy and sell on the same sets up to an error that vanishes by Lemma \ref{l:B1}. Therefore, the difference between the total variation of the convex combinations and the convex combination of the total variations vanishes by Lemma \ref{l:B1} as well. 

To see this, we observe that we can assume without loss of generality after possibly passing to a smaller set $A$ that $\inf_{0\leq u\leq \varrho\wedge T}\hZ^0_u>\bar{c}$ for some $\bar{c}>0$. This follows by the minimum principle for supermartingales. Then, we can choose $j\in\N$ sufficiently large such that the sets $B^c_{1,j}=\{\hZ^0S-\hZ^1\leq\frac{1}{j}\}$ and $B^c_{2,j}=\{\hZ^1-\hZ^0(1-\lambda)S\leq\frac{1}{j}\}$, where $B_{1,j}=\{\hZ^0S-\hZ^1>\frac{1}{j}\}$ and $B_{2,j}=\{\hZ^1-\hZ^0(1-\lambda)S>\frac{1}{j}\}$ are as defined in Lemma \ref{l:B1}, are disjoint on $\{\inf_{0\leq u\leq \varrho\wedge T}\hZ^0_u>\bar{c}\}$. Therefore, we can estimate on $\{\inf_{0\leq u\leq \varrho\wedge T}\hZ^0_u>\bar{c}\}$ that
\begin{align*}
|\hvp^{1,n}|_{\varrho\wedge T}&=\left|\sum_{k=1}^{K_n}\mu_k^n\vp^{1,m^n_k}\right|_{\varrho\wedge T}\\
&=\left|\sum_{k=1}^{K_n}\mu_k^n\left(\vp^{1,m^n_k,\uparrow}-\vp^{1,m^n_k,\downarrow}\right)\right|_{\varrho\wedge T}\\
&=\left|\sum_{k=1}^{K_n}\mu_k^n\left(\mathbbm{1}_{B^c_{1,j}}\sint\vp^{1,m^n_k,\uparrow}-\mathbbm{1}_{B^c_{2,j}}\sint\vp^{1,m^n_k,\downarrow}+\mathbbm{1}_{B_{1,j}}\sint\vp^{1,m^n_k,\uparrow}-\mathbbm{1}_{B_{2,j}}\sint\vp^{1,m^n_k,\downarrow}\right)\right|_{\varrho\wedge T}\\
&\geq\sum_{k=1}^{K_n}\mu_k^n\left(\mathbbm{1}_{B^c_{1,j}}\sint\vp^{1,m^n_k,\uparrow}_{\varrho\wedge T}+\mathbbm{1}_{B^c_{2,j}}\sint\vp^{1,m^n_k,\downarrow}_{\varrho\wedge T}\right)-\sum_{k=1}^{K_n}\mu_k^n\left(\mathbbm{1}_{B_{1,j}}\sint\vp^{1,m^n_k,\uparrow}_{\varrho\wedge T}+\mathbbm{1}_{B_{2,j}}\sint\vp^{1,m^n_k,\downarrow}_{\varrho\wedge T}\right)\\
&= \sum_{k=1}^{K_n}\mu_k^n|\vp^{1,m^n_k}|_{\varrho\wedge T}-2\sum_{k=1}^{K_n}\mu_k^n\left(\mathbbm{1}_{B_{1,j}}\sint\vp^{1,m^n_k,\uparrow}_{\varrho\wedge T}+\mathbbm{1}_{B_{2,j}}\sint\vp^{1,m^n_k,\downarrow}_{\varrho\wedge T}\right).
\end{align*}
Similarly, we also obtain on $\{\inf_{0\leq u\leq \varrho\wedge T}\hZ^0_u>\bar{c}\}$ that
\begin{align*}
|\hvp^{0,n}|_{\varrho\wedge T}&\geq \sum_{k=1}^{K_n}\mu_k^n|\vp^{0,m^n_k}|_{\varrho\wedge T}-2\sum_{k=1}^{K_n}\mu_k^n\left(\mathbbm{1}_{B_{1,j}}\sint\vp^{0,m^n_k,\downarrow}_{\varrho\wedge T}+\mathbbm{1}_{B_{2,j}}\sint\vp^{0,m^n_k,\uparrow}_{\varrho\wedge T}\right).
\end{align*}
Combining both estimates gives on $\{\inf_{0\leq u\leq \varrho\wedge T}\hZ^0_u>\bar{c}\}$ that
$$|\hvp^n|_{\varrho\wedge T}\geq \psi^n_{\varrho\wedge T}-2\sum_{k=1}^{K_n}\mu_k^n\left(\mathbbm{1}_{B_{1,j}}\sint\vp^{1,m^n_k,\uparrow}_{\varrho\wedge T}+\mathbbm{1}_{B_{2,j}}\sint\vp^{1,m^n_k,\downarrow}_{\varrho\wedge T}+\mathbbm{1}_{B_{1,j}}\sint\vp^{0,m^n_k,\downarrow}_{\varrho\wedge T}+\mathbbm{1}_{B_{2,j}}\sint\vp^{0,m^n_k,\uparrow}_{\varrho\wedge T}\right).$$
Since we have
\begin{align*}
&\mathbbm{1}_{B_{1,j}}\sint\vp^{1,n,\uparrow}_T+\mathbbm{1}_{B_{2,j}}\sint\vp^{1,n,\downarrow}_T\xrightarrow{P}0,\\
&\mathbbm{1}_{B_{1,j}}\sint\vp^{0,n,\downarrow}_T+\mathbbm{1}_{B_{2,j}}\sint\vp^{0,n,\uparrow}_T\xrightarrow{P}0.
\end{align*}
by Lemma \ref{l:B1}, this implies that
$|\hvp^n|_{\varrho\wedge T}\xrightarrow{P}\infty$ on $\{\inf_{0\leq u\leq \varrho\wedge T}\hZ^0_u>\bar{c}\}$ and therefore that $\int^T_\tau|d\hvp^n_u| \xrightarrow{P} \infty$ on $\{\inf_{0\leq u\leq \varrho\wedge T}\hZ^0_u>\bar{c}\}\cap A$, as $| \hvp^n|_{\tau\wedge T}\leq c$ for all $n\geq 1$ on $A$.
\end{proof}

After the preparations above, we can now show that $\widehat{\sigma}=\infty$ $P$-a.s. This proves parts 1) and 2) of Theorem \ref{mt3}. Assertions 3) and 4) then follow from Proposition \ref{mt2B}.
\begin{lemma}\label{A:l2A}
Under the assumptions of Theorem \ref{mt3}, we have that $\widehat{\sigma}=\infty$ P-a.s. 

That is, for any maximising sequence $\varphi^n=(\varphi^{0,n}_t, \varphi^{1,n}_t)_{0 \leq t \leq T}\in\cA^\lambda_{adm}(x)$ of trading strategies satisfying
\eqref{ms}, we have that $C:=\conv\{|\varphi^n|_T~;~n \geq 1\}$ is bounded in $L^0(P)$.
\end{lemma}
\begin{proof}
We argue by contradiction and assume that $P(\widehat{\sigma} < \infty) > 0.$ Then there exists by 2) of Lemma \ref{A:l1} a stopping time $\tau$, a set $A\subseteq\{\tau<T\}$ with $P(A)>0,$ a constant $c>0$ and a sequence $(\hvp^n)^\infty_{n=1}$ of convex combinations
$$\hvp^n \in \conv (\bar{\varphi}^n, \bar{\varphi}^{n+1}, \dots)$$
such that we have on $A$ that
\bi
\item[{\bf a)}] $\int^\tau_0 |d\hvp^n_u| \leq c$ for all $n$,
\item[{\bf b)}] $\int^T_\tau|d\hvp^n_u| \to \infty,$ as $n \to \infty$,
\item[{\bf c)}] $|S_t - S_\tau| \leq \frac{\lambda}{3} S_t$ for all $t \in [\tau,T].$
\ei

As we can assume without loss of generality that $\hvp^{1,n}_T=0$, we obtain by combining a) -- c) with the self-financing condition \eqref{eq:sf} under transaction costs similarly as in \eqref{eq:s1} that
\begin{align}
V^{liq}_T(\hvp^n)=\hvp^{0,n}_T&\leq x -\int\limits_{0}^{T} S_u d\hvp^{1,n}_u - \lambda \int\limits_{0}^{T} S_u d\hvp^{1,n,\downarrow}_u \nonumber\\
&=\hvp^{0,n}_\tau+\hvp^{1,n}_\tau S_\tau-\int\limits_{\tau}^{T} (S_u - S_\tau) d\hvp^{1,n}_u - \lambda \int\limits_{\tau}^{T} S_u d\hvp^{1,n,\downarrow}_u\nonumber \\
&\leq \hvp^{0,n}_\tau+\hvp^{1,n}_\tau S_\tau - \frac{2}{3} \lambda \int\limits_{\tau}^{T} S_u d\hvp^{1,n,\downarrow}_u\to-\infty,\quad \text{as $n\to\infty$, on $A$}.\label{A:eq7}
\end{align}
Note that $\hvp^{1,n}_T=0$ implies that $\int^T_\tau d\hvp^{1,n,\downarrow}_u \to \infty,$ as $n \to \infty$, on $A$ by b).

Since $\hvp^n \in\conv (\varphi^n, \varphi^{n+1}, \dots),$ the sequence $(\hvp^n)^\infty_{n=1}$ also has to satisfy
$$U\big(V^{liq}_T (\hvp^n)\big) \stackrel{L^1(P)}{\longrightarrow} U\big( \hg(x)\big).$$
However, this contradicts \eqref{A:eq7} and we therefore have that $P(\widehat{\sigma} < \infty) = 0.$
\end{proof}
\begin{proof}[Proof of Theorem \ref{mt3}]
We only need to prove 2). This immediately implies 1) and 3) and 4) by Proposition \ref{mt2B}. As explained after the statement of Theorem \ref{mt3} on page \pageref{mt3:p1}, the assumptions of the Duality Theorem \ref{mainthm} are satisfied under the assumptions of Theorem \ref{mt3} and by Lemma \ref{l:nc}. This allows us to apply the Duality Theorem \ref{mainthm} to obtain a maximising sequence $\varphi^n=(\varphi^{0,n}_t, \varphi^{1,n}_t)_{0 \leq t \leq T}\in\cA^\lambda_{adm}(x)$ of self-financing and admissible trading strategies and a random variable $\hg=\hg(x) \in L^0 (P; \mathbb{R} \cup \{\infty\})$ such that $E\big[U\big(\hg(x)\big)\big] = u(x)$ and
\begin{align}
V^{liq}_T(\varphi^n) &\stackrel{P}{\longrightarrow} \hg(x),\nonumber\\
U\big(V^{liq}_T (\varphi^n)\big)&\stackrel{L^1(P)}{\longrightarrow} U\big( \hg(x)\big).\label{ms2}
\end{align}
By Lemma \ref{A:l2A}, we then have that $C:=\conv\{|\varphi^n|_T~;~n \geq 1\}$
is bounded in $L^0(P)$. Therefore, there exists a sequence $(\hvp^n)^\infty_{n=1}$ of convex combinations
$$\hvp^n \in \conv (\varphi^n, \varphi^{n+1}, \dots)$$
and a self-financing trading strategy $\hvp=(\hvp^0_t, \hvp^1_t)_{0 \leq t \leq T}$ under transaction costs such that
\begin{align}\label{4:eq2}
P\left[(\hvp^{0,n}_t, \hvp^{1,n}_t) \xrightarrow{n\to \infty}(\hvp^0_t, \hvp^1_t),\,\forall t \in [0,T]\right]=1
\end{align}
by Proposition 3.4 in \cite{CS06} (and its application in the proof of Theorem 3.5 therein). The sequence $(\hvp^n)^\infty_{n=1}$ then also satisfies \eqref{ms2}, which completes the proof.
\ep

\section{A case study: Fractional Brownian Motion and Exponential Utility}

We resume here the theme of (exponential) fractional Brownian motion which was briefly discussed in the introduction. In fact, the challenge posed by this example was an important motivation for the present research.

Fractional Brownian motion has been proposed by B.~Mandelbrot \cite{M63} as a model for stock price processes more than 50 years ago. Until today, this idea poses a number of open problems. From a mathematical point of view, a major difficulty arises from the fact that fractional Brownian motion fails to be a semimartingale (except for the Brownian case $H=\frac{1}{2}$). Tools from stochastic calculus are therefore hard to apply and it is difficult to reconcile this model with the usual no arbitrage theory of mathematical finance. Indeed, it was shown in (\cite{DS94}, Theorem 7.2) that a stochastic process which fails to be a semi-martingale automatically allows for arbitrage (in a sense which was made precise in Theorem 7.2). In the special case of fractional Brownian motion, this was also shown directly by C.~Rogers \cite{R97}. 

One way to avoid this deadlock arising from the violation of the no-arbitrage paradigm is the consideration of proportional transaction costs. The introduction of proportional transaction costs $\lambda$, for arbitrarily small $\lambda > 0,$ makes the arbitrage opportunities disappear. Theorem \ref{mt3} applies perfectly to the case of fractional Brownian motion, for any Hurst index $H \in(0,1).$ As utility function $U$, we may, e.g., choose exponential utility $U(x)=-e^{-x}.$ Hence, we dispose of a duality theory for fractional Brownian motion under transaction costs and, in particular, we may find a shadow price process $\widehat{S}$ which is a semimartingale. 

Let us define the setting more formally. As driver of our model $S$, we fix a standard Brownian motion $(W_t)_{-\infty < t < \infty},$ indexed by the entire real line, in its natural (right continuous, saturated) filtration $(\mathcal{F}_t)_{-\infty < t < \infty}.$ We let the Brownian motion $W$ run from $-\infty$ on in order to apply the elegant integral representation below \eqref{C4} due to Mandelbrot and van Ness; see \cite{N03}. 

We note that the Brownian motion $(W_t)_{0 \leq t \leq T}$, now indexed by $[0,T],$ has the integral representation property with respect to the filtration $(\mathcal{F}_t)_{0 \leq t \leq T}.$ The only difference to the more classical setting, where we consider the filtration $(\mathcal{G}_t)_{0 \leq t \leq T}$ generated by $(W_t)_{0 \leq t \leq T}$ is that $\mathcal{F}_0$ is not trivial anymore. But this causes little trouble. We simply have to do all the arguments conditionally on $\mathcal{F}_0.$

Fix a Hurst parameter $H \! \in \! (0,1)\setminus\{ \frac{1}{2}\}.$ We may define the fractional Brownian motion $(B_t)_{0 \leq t \leq T}=(B^H_t)_{0 \leq t \leq T}$ as

\begin{align}\label{C4}
B_t= C (H) \int^t_{-\infty} \left( (t-s)^{H-\frac{1}{2}} - \left(|s|^{H-\frac{1}{2}}\mathbbm{1}_{(-\infty,0)} \right) \right) dW_s,\quad 0 \leq t \leq T,
\end{align}
where $C(H)$ is some constant which is not relevant in the sequel (see \cite{N03}, section 1.1 or \cite{R97}, formula (1.1)). 

We may further define a non-negative stock price process $S=(S_t)_{0 \leq t \leq T}$ by letting
\begin{align}\label{C5}
S_t= \exp(B_t), \quad \quad 0 \leq t \leq T,
\end{align}
or, slightly more generally,
\begin{align}\label{C5a}
S_t= \exp (\sigma B_t + \mu t), \quad \quad 0 \leq t \leq T,
\end{align}
for some $\sigma > 0$ and $\mu \in \mathbb{R}.$ For the sake of concreteness we stick to \eqref{C5}. We now are in a situation covered by Theorem \ref{mt3}.

As regards the stickiness of $S$, this property (Def. \ref{def:sticky}) of (exponential) fractional Brownian motion has been shown by P.~Guasoni \cite{G06}. We also fix transaction costs $\lambda > 0$ and $U(x)=-e^{-x}$, as well as an initial capital $x \in \mathbb{R},$ e.g. $x=0.$ By Theorem \ref{mt3}, we may find a primal optimizer $\hvp=(\hvp^0_t, \hvp^1_t)_{0 \leq t \leq T},$ a dual optimiser $\widehat{Z}=(\widehat{Z}^0_t, \widehat{Z}^1_t)_{0 \leq t \leq T}$ which is a $\lambda$-consistent price system, as well as a shadow price process $\widehat{S}=\frac{\widehat{Z}^1}{\widehat{Z}^0}.$
From this general theorem, we know that $\widehat{Z}^0$ is a uniquely determined martingale and that $\widehat{Z}^1$ is a local martingale. It seems rather obvious that in the present case \eqref{C5} or \eqref{C5a} the process $\widehat{Z}^1$ is, in fact, also a martingale, but we do not need this result and therefore do not attempt to prove it.

These general and rather innocent looking results have some striking consequences, also outside the realm of mathematical finance. They imply that the fractional Brownian paths may touch the paths of an Itô process in a one-sided way (Theorem \ref{thm} below). 
\vskip10pt
Let us draw some conclusions from Theorem \ref{mt3}.
\begin{lemma}\label{5.3}
In the above setting of exponential fractional Brownian motion the martingale $(\widehat{Z}^0_t)_{0 \leq t \leq T}$ has a representation as
\begin{align}\label{p21}
\widehat{Z}_t^0=\widehat{Z}_0^0\exp\left(-\int^t_0 \ \widehat{\alpha}_udW_u - \frac{1}{2} \int^t_0  \widehat{\alpha}^2_u du\right), \quad 0 \leq t \leq T,
\end{align}
for some $\mathbb{R}$-valued predictable (with respect to the filtration $(\mathcal{F}_t)_{0 \leq t \leq T}$) process $\widehat{\alpha}=(\widehat{\alpha}_t)_{0 \leq t \leq T}$ such that $\int^T_0 \widehat{\alpha}^2_t dt < \infty$ almost surely. 

The process $\widehat{X}=\log (\widehat{S})$ is an Itô process and may be represented as
\begin{align}\label{p21a}
\widehat{X}_t = \widehat{X}_0 + \int^t_0 \left(\widehat{\sigma}_u dW_u + \left(\widehat{\mu}_u - \frac{\widehat{\sigma}^2_u}{2}\right)du\right), \quad 0 \leq t \leq T,
\end{align}
where $\widehat{\sigma}$ and $\widehat{\mu}$ are $\mathbb{R}$-valued predictable processes such that $\int^T_0 \widehat{\sigma}^2_t dt$ as well as $\int^T_0 |\widehat{\mu}_t|dt$ are a.s.~finite. In fact, $\widehat{S} = \exp (\widehat{X})$ is a local martingale under the measure $\widehat{Q}$ defined by $\frac{d\widehat{Q}}{dP}= \widehat{Z}^0_T.$ We therefore have the relation 
\begin{align}\label{p21b}
\widehat{\alpha}_u = \frac{\widehat{\mu}_u}{\widehat{\sigma}_u}, \quad u \in [0,T].
\end{align}
This equality holds $m \otimes P$ almost surely, where $m$ is Lebesgue-measure on $[0,T].$ The equality is defined to hold true in the case when the right hand side is of the form $\frac{0}{0}.$ 
\end{lemma}

\begin{proof}
We know from Theorem \ref{mt3} that $\widehat{Z}^0$ and $\widehat{Z}^1$ are local martingales so that we may apply the martingale representation theorem which implies \eqref{p21}. We deduce that $\widehat{S}=\frac{\widehat{Z}^1}{\widehat{Z}^0}$ as well as $\widehat{X}=\log (\widehat{S})$ are Itô processes which yields a representation of the form \eqref{p21a}. Passing again to $\widehat{S}=\exp(\widehat{X})$ we obtain 
$$\frac{d\widehat{S}_t}{\widehat{S}_t}= \widehat{\sigma}_t dW_t + \widehat{\mu}_t dt,$$
which implies equality \eqref{p21b} by Girsanov and the fact that $\widehat{S}$ is a local martingale under $Q$.
\end{proof}
Before formulating the main result of this section we still need some preparation which also is of some independent interest. 
\begin{lemma}
For $0 < \lambda <1,$ denote by $u^{(\lambda)}(x)$ the corresponding indirect utility function \eqref{tag5}.
Then
\begin{equation}\label{L1a}
u^{(\lambda)} (x)=-f(\lambda)e^{-x}, \quad 0 < \lambda < 1,
\end{equation}
where $f (\lambda)$ is a non-decreasing function taking values in $(0,1]$ and 
\begin{equation}\label{L1}
\lim_{\lambda \searrow 0} f(\lambda)=0.
\end{equation}
\end{lemma}

\begin{proof}
The fact that $u^{(\lambda)}$ is of the form \eqref{L1a} is a well-known scaling property of exponential utility. 

Let us analyze the function $f(\lambda).$ It is obvious that $f(\lambda)$ in non-decreasing and takes its values in $(0,1].$ As regards \eqref{L1}, it follows from \cite{R97} (or the proof of Theorem 7.2 in \cite{DS94}) that we may find, for $\varepsilon >0$ and $M>0$, a simple predictable process $\vt$ of the form 
$$\vt_t= \sum^{N-1}_{i=0} g_i \mathbbm{1}_{\rrbracket \tau_i, \tau_{i+1} \rrbracket}(t)$$
where $g_i \in L^\infty (\Omega, \mathcal{F}_{\tau_i}, P)$ and $0=\tau_0 \leq \tau_1 \leq \dots \leq \tau_N=T$ are stopping times such that, for $S=\exp(B),$
\begin{equation}
(\vt\sint S)_T = \sum^N_{i=0} g_i (S_{\tau_{i+1}} - S_{\tau_i})
\end{equation}
satisfies $(\vt\sint S)_T \geq -1$ almost surely and $P[(\vt\sint S)_T \geq M] > 1 - \varepsilon.$

For $0 < \lambda <1,$ we may $\vt$ interpret also in the setting of transaction costs. More formally: associate to $\vt$ a $\lambda$-self-financing process $\varphi=(\varphi^0, \varphi^1)$ as above starting at $(\varphi^0_0, \varphi^1_0)=(0,0),$ such that $\varphi^1=\vt \mathbbm{1}_{(0,T)}$ and $\varphi^0$ is defined by having equality in \eqref{eq:sf}. Choosing $\lambda >0$ sufficiently small we obtain $\varphi^0_T \geq -2$ almost surely as well as $P[\varphi^0_T \geq M-1] > 1 - \varepsilon.$ This readily shows \eqref{L1}.
\end{proof}

We now can formulate a consequence of the above results on portfolio optimisation which seems remarkable, independently of the above financial applications, as a general result on the pathwise behaviour of fractional Brownian motion: they may touch It\^o processes in a non-trivial way without involving local time or related concepts pertaining to the reflection of Brownian motion.

\bt\label{thm}
Let $(B_t)_{0 \leq t \leq T}$ be fractional Brownian motion with Hurst index $H \in(0,1)\setminus\{\frac{1}{2}\}$ and $\alpha >0$ (which corresponds to $\alpha=-\log (1-\lambda)$ in the above setting of transaction costs).

There is an Itô process $(X_t)_{0 \leq t \leq T}$ such that
\begin{align}\label{L2}
B_t-\alpha \leq X_t \leq B_t, \quad \quad 0 \leq t \leq T,
\end{align}
holds true almost surely.

In addition, $X$ can be constructed in such a way that $(e^{X_{t}})_{0 \leq t \leq T}$ is a local martingale under some measure $Q$ equivalent to $P.$ For $\varepsilon >0,$ we may choose $\alpha >0$ sufficiently small so that the trajectory $(X_t)_{0 \leq t \leq T}$ touches the trajectories $(B_t)_{0 \leq t \leq T}$ as well as the trajectories $(B_t -\alpha)_{0 \leq t \leq T}$ with probability bigger than $1-\varepsilon.$
\et

\begin{proof}
The theorem is a consequence of Theorem \ref{mt3} and Lemma \ref{5.3} where we simply take $X=\widehat{X}.$

We only have to show the last assertion. It translates into the setting of Theorem \ref{mt3} as the statement that, for $\varepsilon > 0,$ there is $\lambda_0 > 0$ such that, for $0 < \lambda < \lambda_0,$ we have with probability bigger than $1-\varepsilon$ that $(\widehat{\varphi}_t)_{0 \leq t \leq T}$ is not constant. Indeed, apart from the trivial case $\widehat{\varphi}_t \equiv (x,0)$ of no trading there must be some buying as well as some selling of the stock, as the investor starts and finishes with zero holdings of stock. As this can only happen if $\widehat{S}_t=S_t$ or $\widehat{S}_t=(1-\lambda)S_t$ respectively, we must have equality in \eqref{L2} for both cases for some $t \in [0,T]$. To show that this case occurs with probability bigger than $1-\varepsilon$, for sufficiently small enough $\alpha >0$, assume to the contrary that there are $\eta > 0$ and arbitrary small $\alpha>0$ such that the optimal trading strategy $\widehat{\varphi}$ remains constant with probability bigger than $\eta.$ This contradicts \eqref{L1} as then we have 
\begin{align*}
u^{\lambda}(0) \leq -\eta.
\end{align*}
\end{proof}

Let us comment on the interpretation of the above theorem. Using the above construction define $\sigma$ and $\tau$ to be the stopping time
\begin{align*}
\sigma= \inf \{ t \in [0,T]: X_t=B_t - \alpha\}, \quad \tau= \inf\{t \in [0,T]: X_t=B_t\},
\end{align*}
which for sufficiently small $\alpha >0,$ satisfies $P[\sigma < \infty]= P[\tau < \infty] > 1 - \varepsilon.$ Here, the equality $P[\sigma < \infty]= P[\tau < \infty]$ follows from the fact that, since we start and end with zero holdings in stock, any position that is bought or sold has to be liquidated before time $T$. We may suppose w.l.o.g.~that $\tau < \sigma$ (the case $\sigma < \tau$ is analogous). Consider the difference process
\begin{align}\label{L3a}
D_t = B_t - X_t, \quad \quad 0 \leq t \leq T,
\end{align}
which, is non-negative and vanishes for $t=\tau.$ We formulate a consequence of the above considerations.

\begin{cor}
On the set $\{\tau < \sigma\}$ we have that $\sigma \leq T$ almost surely, and that the process $(D_t)_{\tau \leq t \leq \sigma}$ starts at zero, remains non-negative and ends at $D_\sigma = \alpha.$ $\square$
\end{cor}

This statement should be compared to the well-known fact, that there are {\it no} stopping times $\tau < \sigma$ such that $P[\tau < T]=P[\sigma \leq T] > 0$ and such that $B_\sigma - B_\tau > \alpha$, almost surely on $\{\tau < T\}.$ Indeed, this follows from the stickiness property (Def. \ref{def:sticky}) of fractional Brownian motion proved by P.~Guasoni (\cite{G06}; compare also \cite{GRS10}). Adding to $B$ the It\^o process $X$ somewhat miraculously changes this behaviour of $B$ drastically as formulated in the above corollary.

\begin{appendix}
\section{An abstract version of the duality theorem}
The basic idea to prove the Duality Theorem \ref{mainthm} under transaction costs is, as in \cite{CS14}, to reduce it to an abstract version of the duality theorem in the frictionless case in \cite{S01}. We provide this abstract version that is what was actually shown in the proof of Theorem 2.2 in \cite{S01} below. It might find other applications as well.

To that end, let $\cC$ be a closed, convex, solid and bounded subset of $L^0_+(P)$ containing the constant $1$, set $\cC(x)=x\,\cC$ for all $x>0$ and $\cC_b(x)=\cup_{n=1}^\infty\{\cC(x+n)-n\}$ for all $x\in\R$. Denote by $\cD$ the polar of $\cC$ in $L^0_+(P)$ given by $\cC^{\circ}=\{h\in L^0_+(P)~|~E[gh]\leq 1\quad\forall g\in\cC\}$ and set $\cD(y)=y\,\cD$ for all $y>0$. Note that, since $1\in\cC$, we have that $E[h]\leq 1$ for all $h\in\cD$. Suppose that $D=\{h\in\cD~|~h>0~\text{and}~E[h]=1\}$ is non-empty and such that $\cD$ is the closed, convex and solid hull of $D$ in $L^0_+(P)$. Denote by $\overline{D}$ the $L^1(P)$-closure of $D$ given by $\overline{D}=\{h\in\cD~|~E[h]=1\}$.

As shown in Theorem 3.2 of \cite{KS99}, the properties of the sets $\cC(x)$ and $\cD(y)$ above are the ones that are needed to establish the duality theory for utility maximisation on the positive half-line. The following theorem presents an extension of this result to utility functions on the whole real line.

\bt\label{mtav}
Under the assumptions above, suppose that $U:\R\to\R$ satisfies the Inada conditions,
has reasonable asymptotic elasticity, i.e.
$AE_{\infty}(U):=\varlimsup\limits_{x\to\infty}\frac{xU'(x)}{U(x)}<1$ and $AE_{-\infty}(U):=\varliminf\limits_{x\to-\infty}\frac{xU'(x)}{U(x)}>1$,
and that
\be
u(x):=\sup_{g\in\cC_U(x)}E[U(g)]<U(\infty)\label{eq:fu:av}
\ee
for some $x\in\R$, where
\begin{multline*}
\cC_U (x) = \big\{ g \in L^0 (P; \mathbb{R} \cup \{\infty\})~ |~ \exists g_n \in \cC_b (x)~\text{such that}\\\text{$U(g_n) \in L^1(P)$ and $U(g_n) \stackrel{L^1(P)}{\longrightarrow} U(g)$} \big\}.
\end{multline*}
Then:
\bi
\item[{\bf 1)}] The primal value function $u$, defined in \eqref{eq:fu:av}, and the dual value function
$$v(y):=\inf_{h\in \overline{D}}E[V(yh)],$$
where $V(y):=\sup_{x\in\R}\{U(x)-xy\}$ for $y>0$ denotes the Legendre transform of $U$, are conjugate, i.e.,
\begin{eqnarray*}u(x)=\inf_{y>0}\{v(y)+xy\},\qquad v(y)=\sup_{x\in\R}\{u(x)-xy\},
\end{eqnarray*}
and continuously differentiable. The functions $u$ and $-v$ are strictly concave and satisfy the Inada conditions
$$\text{$\Lim_{x\to-\infty}u'(x)=\infty,\qquad\Lim_{y\to\infty}v'(y)=\infty,\qquad\Lim_{x\to\infty}u'(x)=0,\qquad\Lim_{y\to0}v'(y)=-\infty$}.$$
The primal value function $u$ has reasonable asymptotic elasticity. 
\item[{\bf 2)}] For $y>0$, the solution $\hh(y)\in \overline{D}$ to the dual problem
\begin{equation}
\textstyle
E\left[V\big(yh\big)\right]\to\min!\label{D1:av}, \qquad{h\in\overline{D}},
\end{equation}
exists, is unique and the map $y\mapsto\hh(y)$ is continuous in variation norm.
\item[{\bf 3)}] For $x\in\R$, the solution $\hg(x)\in\cC_U(x)$ to the primal problem
\begin{equation}
\textstyle
E[U(g)]\to\max!\label{P1:av}, \qquad{g\in\cC_U(x)},
\end{equation}
exists, is unique and given by
\be
\hg(x)=(U')^{-1}\left(\hy(x)\hh\big(\hy(x)\big)\right),\label{eq:dr:1:av}
\ee
where $\hy(x)=u'(x)$.
\item[{\bf 4)}] We have the formulae
$$v'(y)=E\left[\hh(y)V'\big(y\hh(y)\big)\right]\qquad\text{and}\qquad xu'(x)=E\left[\hg(x)U'\big(\hg(x)\big)\right],$$
where we use the convention that $0\cdot\infty=0$, if the random variables are of this form.
\ei
\et
\bp
The proof follows along the same arguments as that of Theorem 2.2 in \cite{S01} after replacing each of the approximating problems (16) in \cite{S01} by its abstract version, i.e. problem (3.4) in \cite{KS99}, and using Theorem 3.2 in \cite{KS99} instead of Theorem 2.2 in \cite{KS99}.
%
%

Indeed, let $\widetilde{S}=(\widetilde{S}_t)_{0 \leq t \leq T}$ be a locally bounded semimartingale price process that admits an equivalent local martingale measure (ELMM) $Q\sim P$ so that the set $\cM^e(\widetilde{S})$ of all ELMM for $\tS$ is non-empty. Denote by $\mathcal{X}(x)$ the set of all non-negative wealth processes starting with initial capital $x$, i.e.
$$X_t=x + \vt \sint \widetilde{S}_t \geq 0, \quad 0 \leq t \leq T,$$
where $\vt \in L(\widetilde{S})$ is an $\widetilde{S}$-integrable predictable process, and by $\mathcal{Y}(y)$ the set of all supermartingale deflators for $\widetilde{S}$, i.e.~non-negative optional strong supermartingales $Y=(Y_t)_{0 \leq t \leq T}$ starting at $Y_0=y$ such that $Y X=(Y_tX_t)_{0 \leq t \leq T}$ is a non-negative supermartingale for all $X \in \mathcal{X}(1).$ Then the abstract sets above correspond to the following sets in \cite{S01}
\begin{align*}
\cC &\triangleq \{ g \in L^0_+(P)~|~\text{$\exists X \in \mathcal{X}(1)$ such that $g \leq X_T$}\},\\
\cC (x)  &\triangleq \{ g \in L^0_+(P)~|~\text{$\exists X \in \mathcal{X}(x)$ such that $g \leq X_T$}\},\quad x>0,\\
\cC_b(x) &\triangleq \cup^\infty_{n=1} \{\cC (x+n) - n\},\\
\mathcal{D} &\triangleq \{Y_T~|~Y \in \mathcal{Y}(1)\},\\
\mathcal{D}(y)& \triangleq \{Y_T~|~Y \in \mathcal{Y}(y)\},\quad y>0,\\
D &\triangleq \left\{\textstyle\frac{dQ}{dP}~\Big|~Q \in \mathcal{M}^e (\tS)\right\},\\
\overline{D} &\triangleq \left\{\textstyle\frac{dQ}{dP}~\Big|~Q \in \mathcal{M}^a (\tS)\right\}.
\end{align*}
Note that $\cC_b(x)$ corresponds to the set of all random variables $g\in L^0(P)$ that are bounded from below and such that there exists $X \in \mathcal{X}_b (x)$ such that $g\leq X_T$, where $\mathcal{X}_b(x)$ is the set of all wealth processes that are uniformly bounded from below, i.e.~there exists some $M > 0$ such that
$$X_t = x + \vt \sint \widetilde{S}_t \geq-M, \quad 0 \leq t \leq T.$$

Conversely, replacing the ``concrete sets'' above in the proof of Theorem 2.2 in \cite{S01} and using the ``abstract version'' of the duality results for utility functions on the positive half-line in Theorem 3.2 of \cite{KS99} instead of Theorem 2.2 in \cite{KS99} with the ``abstract sets'' yields the proof of the abstract version of the theorem. This is clear for all steps of the proof except step 1, step 3 and step 10. 

In step 1, it is used that by part $(iv)$ of Theorem 2.2 in \cite{KS99} the dual optimiser for the utility maximisation problem on the positive half-line can be approximated by the Radon--Nikodym derivatives of an ELMMs. To ensure this in our ``abstract setting'', one has by Proposition 3.2 in \cite{KS99} to use that the set $\cD$ is the closed, convex and solid hull of $D$ in $L^0_+(P)$ and that $D$ is closed under countable convex combinations. This follows immediately from the assumption that $\mathcal{D}$ is convex and closed in probability and an application of the the monotone convergence theorem.

Step 3 and step 10 show in addition dynamic properties of the primal and dual optimiser that we do not assert and therefore do not need to prove here.
\ep

Applying the abstract duality theorem above to portfolio optimisation under transaction costs then allows us to prove Theorem \ref{mainthm}.

\begin{proof}[Proof of Theorem \ref{mainthm}]
We begin by recalling some of the definitions for portfolio optimisation under transaction costs for utility functions on the positive half line from \cite{CS14}. 

For $x >0$, we denote by $\mathcal{A}^{\lambda} (x)$ the set of all self-financing trading strategies $\varphi=(\varphi^0_t, \varphi^1_t)_{0 \leq t \leq T}$ under transaction costs starting with initial endowment $(\varphi^0_0, \varphi^1_0)=(x,0)$ that are \emph{$0$-admissible}, i.e.~$V^{liq}_t(\varphi) \geq 0$ for all $t \in [0,T].$ The set $\mathcal{B}^{\lambda}(y)$ of all \emph{optional strong supermartingale deflators} consists of all pairs of non-negative optional strong supermartingales $Y=(Y^0_t, Y^1_t)_{0 \leq t \leq T}$ such that $Y^0_0=y,$ $Y^1=Y^0 \widetilde{S}$ for some $[(1-\lambda)S,S]$-valued process $\widetilde{S}=(\widetilde{S}_t)_{0 \leq t \leq T}$ and $Y^0(\varphi^0 + \varphi^1 \widetilde{S})= Y^0 \varphi^0 + Y^1 \varphi^1$ is a non-negative optional strong supermartingale for all $\varphi \in \mathcal{A} (1).$ Note that $\mathcal{Z}^{\lambda}_e \subseteq  \mathcal{Z}^{\lambda}_a \subseteq \mathcal{B}^\lambda (1).$

We define the following sets
\begin{align*}
\cC^\lambda&= \cC^{\lambda} (1) = \{ V^{liq}_T (\varphi)~|~\varphi \in \mathcal{A}^\lambda (1)\},\\
\cC^\lambda (x)&= \{ V^{liq}_T (\varphi)~|~\varphi \in \mathcal{A}^{\lambda} (x)\} , \quad x > 0,\\
\mathcal{D}^{\lambda}&= \mathcal{D}^{\lambda} (1)=\{Y^0_T~|~Y \in \mathcal{B}^{\lambda} (1)\},\\
\mathcal{D}^{\lambda} (y)&= \{Y^0_T~|~Y \in \mathcal{B} (y) \} = y \mathcal{D}^{\lambda}, \quad y > 0,\\
D^{\lambda} &= \{ Z^0_T~|~Z \in \mathcal{Z}^\lambda_e \},\\
\overline{D}^{\lambda}&= \{Z^0_T~|~Z \in \mathcal{Z}^\lambda_a\}. 
\end{align*}

Under the assumptions of Theorem \ref{mainthm} we have by Lemma A.1 in \cite{CS14} that $\cC^\lambda$ is a closed, convex and bounded subset of $L^0_+(P)$ containing the constant $1$, that $\mathcal{D}^\lambda$ coincides with the polar $(\cC^{\lambda})^\circ$ of $\cC^\lambda$ in $L^0_+(P)$ and that $\mathcal{D}^\lambda$ is the closed, convex and solid hull of $D^{\lambda}$ in $L^0_+(P).$

In order to deduce the Duality Theorem \ref{mainthm} by applying the abstract version (Theorem \ref{mtav}) for $\cC=\cC^{\lambda},$ $\mathcal{D}=\mathcal{D}^{\lambda},$ $D=D^{\lambda}$ and $\overline{D}=\overline{D}^{\lambda}$ we therefore only need to verify that 
\begin{align}
&D^{\lambda}=\{ h \in \mathcal{D}^{\lambda}~|~\text{$h >0$ and $E[h]=1$}\}, \label{A:eq2}\\
&\overline{D}^{\lambda}=\{h \in \mathcal{D}^{\lambda}~ |~E[h]=1\}.\label{A:eq3}
\end{align}

We begin with \eqref{A:eq3}. Recall that by the definition of $\mathcal{D}^{\lambda}$ there exists $\overline{Y} =(\overline{Y}^0_t,\overline{Y}^1_t)_{0\leq t\leq T}\in \mathcal{B}^{\lambda}(1)$ such that $\overline{Y}^0_T=h.$ Since $\overline{Y}^0=(\overline{Y}^0_t)_{0 \leq t \leq T}$ is a non-negative optional strong supermartingale starting at $\overline{Y}^0_0=1$, the condition $E[\overline{Y}^0_T]=E[h]=1$ implies that $\overline{Y}^0$ is a true martingale and hence c\`adl\`ag. To see the local martingale property of $\overline{Y}^1=(\overline{Y}^1_t)_{0 \leq t \leq T}$, we need to use the local boundedness of $S=(S_t)_{0 \leq t \leq T}.$ Let $(\tau_n)^\infty_{n=1}$ be a localising sequence of stopping times tending stationarily to $T$ such that $\sup_{0 \leq t \leq T} S^{\tau_n}_t \leq n $ on $\{S_0\leq n\}$. Since $\overline{Y}^1$ is a non-negative optional strong supermartingale, we only need to show that $E[\overline{Y}^1_{\tau_n}\mathbbm{1}_{\{S_0\leq n\}}] \geq E[\overline{Y}^1_0\mathbbm{1}_{\{S_0\leq n\}}]$ to establish the local martingale property of $\overline{Y}^1$ with localising sequence $(\sigma_n)^\infty_{n=1}$ of stopping times given by $\sigma_n=\tau_n\mathbbm{1}_{\{S_0\leq n\}}$.

For this, consider, for $m\geq n$, the self-financing trading strategy $\varphi^m=(\varphi^{0,m}_t, \varphi^{1,m}_t)_{0 \leq t \leq T}$ under transaction costs that starts at $\varphi^m_0=(1,0),$ sells $\frac{1}{m}$ shares of stock immediately after time $0$ on $\{S_0\leq n\}$ and, if $\tau_{m} < T,$ buys them back again at time $\tau_m.$ That is $\varphi^{1,m}= \left(- \frac{1}{m} \mathbbm{1}_{\rrbracket 0,T \rrbracket} + \frac{1}{m} \mathbbm{1}_{\rrbracket \tau_m, T \rrbracket}\right)\mathbbm{1}_{\{S_0\leq n\}}$ and  $\varphi^{0,m}=1 + \left(\frac{1}{m} (1-\lambda) S_0 \mathbbm{1}_{\rrbracket 0,T \rrbracket} - \frac{1}{m} S_{\tau_m} \mathbbm{1}_{\rrbracket \tau_m,T \rrbracket}\right)\mathbbm{1}_{\{S_0\leq n\}}$. The liquidation value of this strategy is given by
\begin{align*}
\textstyle V^{liq}_t(\varphi^m)=1+\left(\frac{1}{m}(1-\lambda) S_0 - \frac{1}{m} S_{\tau_m \wedge t}\right)\mathbbm{1}_{\{S_0\leq n\}} \geq 0, \quad 0\leq t\leq T.
\end{align*}
Therefore $\varphi^m$ is $0$-admissible and $\overline{Y}^0 \varphi^{0,m} + \overline{Y}^1 \varphi^{1,m}$ is an optional strong supermartingale so that
\begin{align*}
1 &\geq E\Big[\overline{Y}^0_{0+} \varphi^{0,m}_{0+} + \overline{Y}^1_{0+} \varphi^{1,m}_{0+}\Big]\\
&= \textstyle E\Big[\overline{Y}^0_0 \big(1 + \frac{1}{m} (1-\lambda) S_0\mathbbm{1}_{\{S_0\leq n\}}\big) - \frac{1}{m} \overline{Y}^1_{0+}\mathbbm{1}_{\{S_0\leq n\}}\Big]\\
&\geq E\Big[\overline{Y}^0_{\tau_n} \varphi^{0,m}_{\tau_n} + \overline{Y}^1_{\tau_n} \varphi^{1,m}_{\tau_n}\Big]\\
&\geq\textstyle E\left[\left(\overline{Y}^0_{\tau_n} \varphi^{0,m}_{\tau_n} + \overline{Y}^1_{\tau_n} \varphi^{1,m}_{\tau_n}\right) \mathbbm{1}_{\{\tau_m=T\}}+\overline{Y}^0_{\tau_n} V^{liq}_{\tau_n}(\varphi^m)\mathbbm{1}_{\{\tau_m < T\}}\right]\\
&=\textstyle E\left[\left(\overline{Y}^0_{\tau_n} \left(1+\frac{1}{m}(1-\lambda)S_0\mathbbm{1}_{\{S_0\leq n\}}\right)- \frac{1}{m} \overline{Y}^1_{\tau_n}\mathbbm{1}_{\{S_0\leq n\}}\right) \mathbbm{1}_{\{\tau_m=T\}}\right]\\
&\textstyle\quad+ E\Big[\overline{Y}^0_{\tau_n}\big(1+\frac{1}{m}(1-\lambda)S_0\mathbbm{1}_{\{S_0\leq n\}} - \frac{1}{m} S_{\tau_n}\mathbbm{1}_{\{S_0\leq n\}}\big) \mathbbm{1}_{\{\tau_m < T\}}\Big].
\end{align*}
By the martingale property of $\overline{Y}^0$ this implies
\begin{align}\label{A:eq4}
\textstyle-\frac{1}{m} E\big[\overline{Y}^1_{0^+}\mathbbm{1}_{\{S_0\leq n\}}\big] \geq \textstyle-\frac{1}{m} E\big[\overline{Y}^1_{\tau_n}\mathbbm{1}_{\{S_0\leq n\}} \mathbbm{1}_{\{\tau_m=T\}}\big]-\frac{1}{m} E\big[\overline{Y}^0_{\tau_n} (1-\lambda)S_{\tau_n}\mathbbm{1}_{\{S_0\leq n\}}\mathbbm{1}_{\{\tau_m < T\}}\big]
\end{align}
and therefore
\begin{align*}
E\big[\overline{Y}^1_{\tau_n}\mathbbm{1}_{\{S_0\leq n\}}\big] \geq E\big[\overline{Y}^1_{0^+}\mathbbm{1}_{\{S_0\leq n\}}\big]
\end{align*}
after multiplying both sides of \eqref{A:eq4} with $m$ and then sending $m$ to infinity, where we use that $P(\tau_m < T) \to 0,$ as $m  \to \infty.$ As $\overline{Y}^0=(\overline{Y}^0_t)_{0 \leq t \leq T}$ and $S=(S_t)_{0 \leq t \leq T}$ are both c\`adl\`ag, we can modify $\overline{Y}^1=(\overline{Y}^1_t)_{0 \leq t \leq T}$ at time $0$ by setting $\overline{Y}^1_0=\overline{Y}^1_{0+}$ to obtain that $\overline{Y}=(\overline{Y}^0_t, \overline{Y}^1_t)_{0 \leq t \leq T}$ is a pair consisting of a martingale $\overline{Y}^0$ and a local martingale $\overline{Y}^1$ such that there exists an $[(1-\lambda)S,S]$-valued process such that $\overline{Y}^1=\overline{Y}^0 \bar{S}$. So we get that there exists $\overline{Y}=(\overline{Y}^0, \overline{Y}^1) \in \mathcal{Z}^{\lambda}_a$ such that $\overline{Y}^0_T=h$ and therefore \eqref{A:eq3}. If $\overline{Y}^0_T=h > 0, $ then $\overline{Y}=(\overline{Y}^0, \overline{Y}^1) \in \mathcal{Z}^\lambda_e$, which proves \eqref{A:eq2}.
\ep
The following auxiliary result was used in the proof of Lemma \ref{l:B1}.
\bl\label{l:A}
Under the assumptions of Theorem \ref{mtav}, let $(g_n)_{n=1}^\infty$ be any sequence of random variables in $\cC_b(x)$ satisfying $U(g_n)\xrightarrow{L^1(P)}U\big(\hg(x)\big)$. Then $\hh\big(\hy(x)\big)g_n\xrightarrow{L^1(P)}\hh\big(\hy(x)\big)\hg(x)$.
\el
\bp
Since $U'$ is non-negative and decreasing, we can estimate
$$\Big(U(g_n)-U\big(\hg(x)\big)\Big)^-\geq U'\big(\hg(x)\big)\big(g_n-\hg(x)\big)^-.$$
Together with the $L^1$-convergence of $U(g_n)$ to $U\big(\hg(x)\big)$, this implies that
$$\Big(U'\big(\hg(x)\big)\big(g_n-\hg(x)\big)^-\Big)_{n=1}^\infty$$ is uniformly integrable and hence that
$$U'\big(\hg(x)\big)\big(g_n-\hg(x)\big)^-\xrightarrow{L^1(P)}0,$$
since $U(g_n)\xrightarrow{L^1(P)}U\big(\hg(x)\big)$ yields that $g_n\xrightarrow{P}\hg(x)\in L^0 (P; \mathbb{R} \cup \{\infty\})$ by the strict monotonicity of $U$. Therefore, we obtain that
\be
\lim_{n\to\infty}E\left[U'\big(\hg(x)\big)\big(g_n-\hg(x)\big)\right]\geq0\label{l:A:eq1}
\ee
by the generalised version of Fatou's lemma. By parts 3) and 4) of Theorem \ref{mtav}, we have that $U'\big(\hg(x)\big)=\hy(x)\hh\big(\hy(x)\big)\in\hy(x)\overline{D}$ and 
\be
E\left[U'\big(\hg(x)\big)\big(g_n-\hg(x)\big)\right]=\hy(x)E\left[\hh\big(\hy(x)\big)\big(g_n-\hg(x)\big)\right]\leq 0\label{l:A:eq2}.
\ee
Combining \eqref{l:A:eq1} and \eqref{l:A:eq2} gives $
\lim_{n\to\infty}E\left[U'\big(\hg(x)\big)\big(g_n-\hg(x)\big)\right]=0
$
and therefore that
$$U'\big(\hg(x)\big)\big(g_n-\hg(x)\big)^+\xrightarrow{L^1(P)}0.$$
The convergence $\hh\big(\hy(x)\big)g_n\xrightarrow{L^1(P)}\hh\big(\hy(x)\big)\hg(x)$
then follows, since $U'\big(\hg(x)\big)=\hy(x)\hh\big(\hy(x)\big)$.
\ep
\end{appendix}
\bibliography{SPEU-2016-02-08}
\bibliographystyle{abbrv}
\end{document}